\documentclass{article}
\usepackage[utf8]{inputenc}
\usepackage{amsmath}
\usepackage{amssymb}
\usepackage[normalem]{ulem}
\usepackage{cancel}
\usepackage[dvipsnames]{xcolor}
\usepackage{stackengine}
\usepackage{graphicx,caption}
\usepackage{multirow}
\usepackage{pgfplotstable}
\usepackage{physics}
\usepackage{csquotes}
\usepackage[english]{babel}
\usepackage{amsthm}
\usepackage{mathtools}
\usepackage{algorithm}
\usepackage[skip=2pt,font=scriptsize]{caption}
\usepackage{subcaption}

\pgfplotsset{compat=1.18}

\renewcommand\part[1]{\vspace{.10in}\textbf{(#1)}}

\newtheorem{claim}{Claim}
\newtheorem{theorem}{Theorem}

\newtheorem{corollary}{Corollary}

\definecolor{myurlcolor}{rgb}{0,0,0.7}
\definecolor{myrefcolor}{rgb}{0.8,0,0}
\usepackage[unicode=true,pdfusetitle, bookmarks=false,bookmarksnumbered=false,
bookmarksopen=false, breaklinks=false,pdfborder={0 0 0},backref=false,
colorlinks=true, linkcolor=myrefcolor,citecolor=myurlcolor,urlcolor=myurlcolor]
{hyperref}

\usepackage{hyperref}
\DeclareMathOperator*{\argmax}{arg\,max}

\usepackage[
backend=bibtex,
sorting=none
]{biblatex} 
\newtheorem{definition}{Definition}

\title{Compilation of product-formula Hamiltonian simulation via reinforcement learning}
\date{}
\author{\small Lea M. Trenkwalder$^1$, Eleanor Scerri$^2$\small, Thomas E. O'Brien$^{2,5}$, and Vedran Dunjko$^{3,4}$}
\date{\it{ \small$^1$Institute for Theoretical Physics, University of Innsbruck, 6020 Innsbruck, Austria\\ $^2$Instituut-Lorentz, Universiteit Leiden, P.O. Box 9506, 2300 RA Leiden, The Netherlands\\ 
$^3$ Applied Quantum Alogorithms, Universität Leidden\\
$^4$ LIACS, Universiteit Leiden, P.O. Box 9512, 2300 RA Leiden, The Netherlands\\$^5$Google Quantum AI, 80636 Munich, Germany}\\[2ex] \large \today}

\begin{document}
\maketitle

\begin{abstract}Hamiltonian simulation is believed to be one of the first tasks where quantum computers can yield a quantum advantage.
One of the most popular methods of Hamiltonian simulation is Trotterization, which makes use of the approximation $e^{i\sum_jA_j}\sim \prod_je^{iA_j}$ and higher-order corrections thereto.
However, this leaves open the question of the order of operations (i.e. the order of the product over $j$, which is known to affect the quality of approximation).
In some cases this order is fixed by the desire to minimise the error of approximation; when it is not the case, we propose that the order can be chosen to optimize compilation to a native quantum architecture.
This presents a new compilation problem --- order-agnostic quantum circuit compilation --- which we prove is NP-hard in the worst case.
In lieu of an easily-computable exact solution, we turn to methods of heuristic optimization of compilation.
We focus on reinforcement learning due to the sequential nature of the compilation task, comparing it to simulated annealing and Monte Carlo tree search.
While two of the methods outperform a naive heuristic, reinforcement learning clearly outperforms all others, with a gain of around 12\% with respect to the second-best method and of around 50\% compared to the naive heuristic in terms of the gate count.
We further test the ability of RL to generalize across instances of the compilation problem, and find that a single learner is able to solve entire problem families.
This demonstrates the ability of machine learning techniques to provide assistance in an order-agnostic quantum compilation task.
\end{abstract}

\section{Introduction}
\label{sec:into}

The computational speed-ups promised by large-scale quantum computers for solving problems such as factoring \cite{shor1997} and approximate optimization \cite{farhi2014} have led to significant advancements in both experimental and theoretical aspects of quantum computation. However, realizing the original algorithms proposed for quantum computers requires protection from environmental noise using quantum error correction (QEC), which introduces a daunting overhead in terms of the required number of qubits. Efficient and effective compilation of quantum circuits is crucial to avoid unnecessary overhead and maximize device capability. Despite being a relatively young field, several approaches have been studied for synthesizing gates and compiling quantum algorithms \cite{maslov2008, fowler2011, chen2013, chong2017, nam2018, heyfron2018, gokhale2019, duncan2020, camps2020}, employing techniques from machine learning \cite{ferrari2018, cincio2018} or more specifically, reinforcement learning \cite{Moro2021,Zhang2020}, as well as automated planning \cite{venturelli2018, booth2018} and quantum-assisted methods \cite{khatri2019, jones2020}. These efforts aim to enhance the efficiency of quantum circuit compilation and overcome the limitations imposed by current quantum devices.

One of the most popular quantum algorithms is Hamiltonian simulation: the implementation of  $e^{i H t}$ as a unitary on a quantum device for a known Hamiltonian $H$~\cite{lloyd1996universal}.
This problem has attracted a large amount of attention over the years, and various asymptotically optimal or near-optimal Hamiltonian simulation algorithms are known~\cite{childs2012hamiltonian,childs2019theory,low2017optimal,kivlichan2019improved}.
One standard method for Hamiltonian simulation is via product formula methods, most notably those based on the Trotter-Suzuki decomposition~\cite{lloyd1996universal,suzuki1991general,whitfield2011,hastings2014improving,babbush2015chemical,childs2019theory,kivlichan2019improved,heyl2019quantum,zhao2022making,tran2020destructive,hagan2022composite,rendon2022improved}, or randomized variants thereof~\cite{campbell2018random,wan2022randomized}.
These are generally based on the approximation
\begin{equation}\label{eq:product_formula}
    \exp(iHt)=\exp\Big(i\sum_j\alpha_jP_j\Big)\sim \prod_j\exp\Big(i\alpha_jP_j\Big),
\end{equation}
where we have assumed $H$ is written as a linear combination of some operators $P_j$ for which device implementations of $e^{i\theta P_j}$ are known.
For example, one can implement rotation by an arbitrary Pauli operator $P_j\in\mathbb{P}^Q:=\{I,X,Y,Z\}^{\otimes Q}$ on $Q$ qubits by a combination of single-qubit rotations and CNOT gates~\cite{lloyd1996universal,whitfield2011}.
However, when $P_j$ acts on a large number of qubits (which is the case for, e.g., quantum chemistry), the resulting circuit can be dominated by the contribution from `strings' of CNOT gates.
In a naive sequential compilation, this can lead to a significant overhead in gate cost and circuit depth.
Though the additional gates here are Clifford and thus relatively cheap in a fault-tolerant cost model~\cite{kivlichan2019improved}, they are not inconsequential, and significant gains can be obtained by even heuristic optimization of compilation~\cite{hastings2014improving}.

The above optimization is made simultaneously more powerful and more difficult as we have not declared the order for the product over $j$ in Eq.~\eqref{eq:product_formula}.
This order is relevant, as the individual terms no longer commute, and in some cases, significant gains in Trotter error~\cite{tran2020destructive} or non-Clifford gate count~\cite{mukhopadhyay2022synthesizing} can be achieved via proper ordering of terms.
But when this is not the case, we have the freedom to optimize term ordering to minimise circuit depths or gate-counts.
This presents a problem of `order-agnostic circuit compilation', where our compiler must choose the order of operation in addition to the implementation of each operation.
Two natural questions then emerge: how difficult is this problem to solve, and what gain is obtained from an optimal solution?
Reinforcement learning, as a subfield of machine learning~\cite{bookSuttonBarto,bookSuttonBarto2nd}, offers a paradigm for training learning algorithms -- so-called learning agents in the ML vernacular -- to make sequential choices of actions maximizing a given figure of merit.

Over the last few years we have witnessed a rise in the use of reinforcement learning to solve problems in several fields in quantum computing, including combinatorial optimization \cite{mckiernan2019, khairy2019, beloborodov2021}, state preparation \cite{bukov2018, zhang2019, alam2019, mackeprang2020, yao2020} as well as error correction \cite{fosel2018, nautrup2019, colomer2020}.
Due to the sequential nature of program execution, reinforcement learning lends itself naturally to compilation tasks~\cite{Zhang2020, Moro2021}, making it an obvious choice to target the order-agnostic compilation of Trotterized quantum circuits.

In this work, we investigate restricted gate set synthesis with multi-qubit Pauli gate sequences consisting of single-qubit Clifford gates and nearest neighbor CNOT and SWAP gates. First, we prove that already a simpler subset of this problem is in the worst case $\mathsf{NP}$-hard. This motivates our choice to use heuristic optimization and data-driven methods, namely reinforcement learning (RL), as a tool to synthesize the gate sets. We demonstrate that an RL agent successfully solves instances with up to 7 qubits, beating one of the most common approaches used to synthesize such Pauli gates in terms of gate sequence lengths. We compare the performance of the RL agent to simulated annealing and Monte Carlo tree search, observing that, once again, the former generates solutions significantly shorter in mapping gate count. In contrast to the other methods, RL allows solving multiple problem instances at once without any learning overhead due to the generalization capabilities of neural networks.  

This paper is divided as follows. In the next section, Sec.~\ref{sec:gsc}, we introduce the gate set conversion (GSC) problem, discussing two different ways of formulating a solution in terms of native and mapping gates that implement the target gate product in Sec.~\ref{sec:sim} and Sec.~\ref{sec:ind}. We then show that problems of this type are $\mathsf{NP}$-hard in Sec.~\ref{sec:complexity}, by looking specifically at an application using Pauli gate sets. We discuss the GSC problem as a reinforcement learning task in Sec.~\ref{sec:rlgsc}. Before we present the result on various instances of GSC with 4 up to 7 qubits and gate sets with 8 to 16 gates in Sec.~\ref{sec:perform} and present the length of the shortest found solutions in Sec.~\ref{sec:shortest}. The used RL method is further compared to simulated annealing and Monte Carlo tree search in Sec.~\ref{sec:comparison}. Leveraging the generalization capabilities of the employed RL method, we present the results for a single RL agent trained on up to a thousand different target gate sets in Sec.~\ref{sec:converter}.

\section{Gate set conversion}
\label{sec:gsc}

When implementing an algorithm, hardware- or other restrictions may necessitate the conversion of a set of needed gates to a set of available gates. Suppose we are tasked with simulating the following time evolution on a specific hardware:

\begin{align}
\label{eq:timeevo}
 e^{-i\tau \sum_j \alpha_j P_j} \sim \prod_j \exp \left(-i \tau \alpha_j \tilde{t}_j\right)~=\prod_j t_j,
\end{align}
where $\tau$ is the evolution time. The operators $\tilde{t}_j$ are the \textit{target} operators and $t_j=\exp \left(-i \tau \alpha_j \tilde{t}_j\right)$ are target gates forming a set $T$ (which may be part of a larger set which we shall refer to as $T^u$, standing for T-universal) which we are tasked to implement. Due to, e.g., hardware restrictions, however, this may not always be possible. It may only be possible to implement products of \textit{native} gates $n$ (the set of which shall be referred to as $N^u$), and, furthermore, we have access to an additional \textit{mapping} gate set (which we shall refer to as  $M^u$) whose elements or products thereof map native to target gates (and vice versa). Additionally, we assumed $\tau h_j$ to be small enough such that the error in Eq.~\eqref{eq:timeevo} is negligible and the order of target gates in the product is irrelevant. The challenge, as defined in Definition~\ref{def:GSC}, is thus to find the shortest product of native and mapping gates that implements the product of target gate sets. This corresponds to determining a sequence of operators that maps all elements in $T$ to elements in $N$. When such a sequence is obtained, the target set $T$ is said to be resolved. 

\begin{definition}[Gate Set Conversion Problem (GSC)]
\label{def:GSC}
Let $(T,N,M)$ be a tuple, where $T$ is the target set, $N$ is the native set, and $M$ is the mapping gate set. Given an instance of GSC $(T,N,M)$, find the shortest sequence of $(m_1,...,m_{k})$ that resolves $T$, where $m_i \in M$ for $1\leq i\leq k$ and $k$ is some integer. 
\end{definition}

An explicit example of this problem would be simulating evolution under a Hamiltonian $H = \sum_j \alpha_j P_j$, where $P_j \in \mathbb{P}^Q = \{ I, X, Y, Z \}^{\otimes Q}$ are Pauli strings and $Q$ is the number of qubits, and $\alpha_j$ are real values. Thus the goal is to simulate the time evolution according to Eq.~\eqref{eq:timeevo} with $\tilde{t_j}=P_j$. In this example, native gates might have the form $\exp \left(i t~\alpha_{j} P'_{j}\right)$ where $P'_j \in \mathbb{P}^Q$.

Although the methods we study are more general, for concreteness we will focus on the cases where the target gate sets $T$ shall be subsets of $T^u  = \mathbb{P}^Q$, whereas the (smaller) native gate set $N^c$ will consist of Pauli strings of the form $I^{\otimes q} \otimes Z \otimes I^{\otimes Q-q-1}$ for some $q \in \{ 0, 1, ..., Q-1 \}$, or $I^{\otimes q} \otimes Z \otimes Z \otimes I^{\otimes Q-q-2}$ for some $q \in \{ 0, 1, ..., Q-2 \}$. As our mapping gate set, we chose $M^u = \{H, S, CNOT, SWAP \}$, which allows us to construct a naive solution as described in Sec.~\ref{sec:ind}. Note, that we have chosen to consider Pauli operators as our target and native gate sets, as opposed to their exponentials as in Eq.~\eqref{eq:timeevo}, since for any $P_j \in \mathbb{P}^Q$ and product $m$ of elements of $M^u$, $m \exp(i \alpha_j P_j) m^\dagger = \exp(i \alpha_j \tilde{P}_j) \Leftrightarrow m^\dagger P_j m = \tilde{P}_j$ ($\tilde{P}_j \in \mathbb{P}^Q$).

\subsection{Individual conversion}
\label{sec:ind}

To tackle the GSC problem, a straightforward approach is to individually map all the target gates to native gates. An example of this is illustrated in Fig~\ref{fig:sin}, where the native gates are rotation gates $R_z (\theta)$. The individually mapped target gates are then concatenated in the circuit in order to implement a gate of the form Eq.~\eqref{eq:timeevo} (once again assuming that the error from first-order Trotterization is negligible). This means that the product of target gates $t_{j_1} \ldots t_{j_{|T|}}$ can now be implemented by a product of native and mapping gates as

\begin{equation}\label{eq:sin}
\prod^{|T|}_{l=1} t_{j_l} = \prod^{|T|}_{l=1} \left[ (\prod^{i_{l}}_{n=1} v_{j_{l_n}}) n_{k_l} (\prod v_{j_{l_n}})^\dag \right]~,
\end{equation}
where $v_{j_{l_n}} \in V^c$, and $(\prod^{i_{l}}_{n=1}  v_{j_{l_n}})^\dag  t_{j_l}(\prod v_{j_{l_n}}) = n_{k_l} \in N^c$ $\forall l \in \{1, \ldots |T|\}$. We shall henceforth refer to a solution of this form as a \textit{individual} solution.

Strategies for specific target, native, and mapping gate sets exist. For example, unitary gates of the form $\exp \left(i \frac{\theta}{2} P_{j}\right)$ in Eq.~\eqref{eq:timeevo} can be decomposed to rotations $R_z (\theta)$ about the z-axis using CNOT, Hadamard (H) and Phase (S) gates \cite{nielsen2011, whitfield2011, gui2020}. Two strategies, illustrated in Fig~\ref{fig:sin2ways}, rely on computing the parity using the CNOT `cascades' and the necessary basis rotations using the single qubit gates. Since we consider devices with linear connectivity as a restriction, we will focus on the `ladder' strategy, illustrated in Fig.~\ref{fig:sin2ways}b. Henceforth, this approach shall be referred to as the \textit{naive} strategy.
%


\begin{figure}
  \begin{minipage}[htbp]{0.45\textwidth} 
    \includegraphics[width=\linewidth]{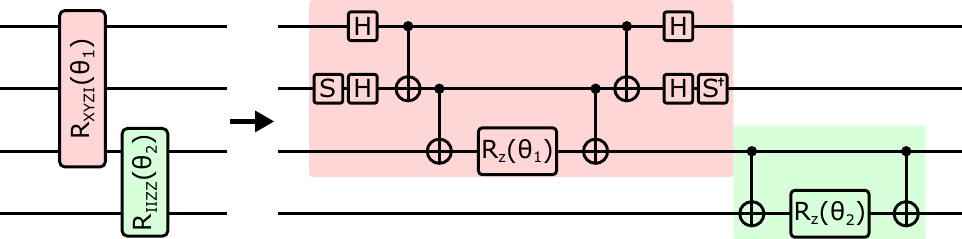}
\caption{Transforming the target gates individually using the `ladder' variant of the transformation illustrated in Fig.~\ref{fig:sin2ways}. Here $R_{XYZI}(\theta_1) = \exp \left(-i \frac{\theta_1}{2} X_1 Y_2 Z_3 \right)$ and $R_{IIZZ}(\theta_2) = \exp \left(-i \frac{\theta_2}{2} Z_3 Z_4 \right)$.}
\label{fig:sin}
  \end{minipage}%
  \hfill
    \begin{minipage}[htbp]{0.5\textwidth}
\includegraphics[width=\linewidth]{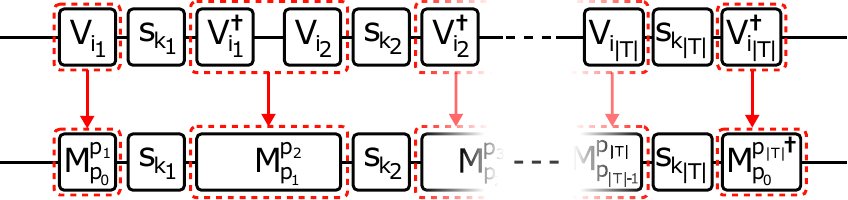}
\caption{Converting a individual solution to a simultaneous one by combining mapping operator products accordingly. Here $V_{i_l} = \prod^{i_l}_{n=1} v_{j_{l_n}}$, and $M^{p_l}_{p_{l-1}} = \prod^{p_l}_{n=p_{l-1}+1} m_n$. Due to cancellations, the final gate product $V_{i_{|T|}}$ can be expressed as a product $\prod^{|T|}_{l=1} M^{p_l}_{p_{l-1}} = M^{p_{|T|}}_{p_0}$.}
\label{fig:sinseq}
  \end{minipage}
\end{figure}


\begin{figure}[t!]
\centering
 \includegraphics[width=0.4\linewidth]{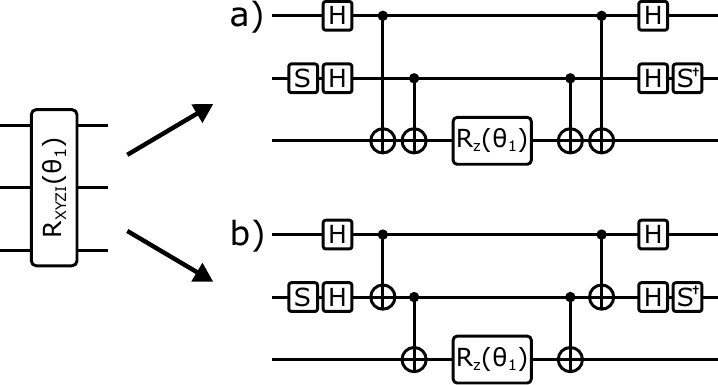}
 \caption{Two ways of decomposing the gate $R_{XYZI}(\theta_1) = \exp \left(-i \frac{\theta_1}{2} X_1 Y_2 Z_3 \right)$, both yielding the same result. While the procedure in a) requires connectivity of all registers to the register the z-rotation is implemented on, the `ladder' strategy in b) only requires connections to the nearest register.}
\label{fig:sin2ways}
\end{figure}

\subsection{Simultaneous conversion}
\label{sec:sim}

\begin{figure}[t!]
\centering
\includegraphics[width=.4\linewidth]{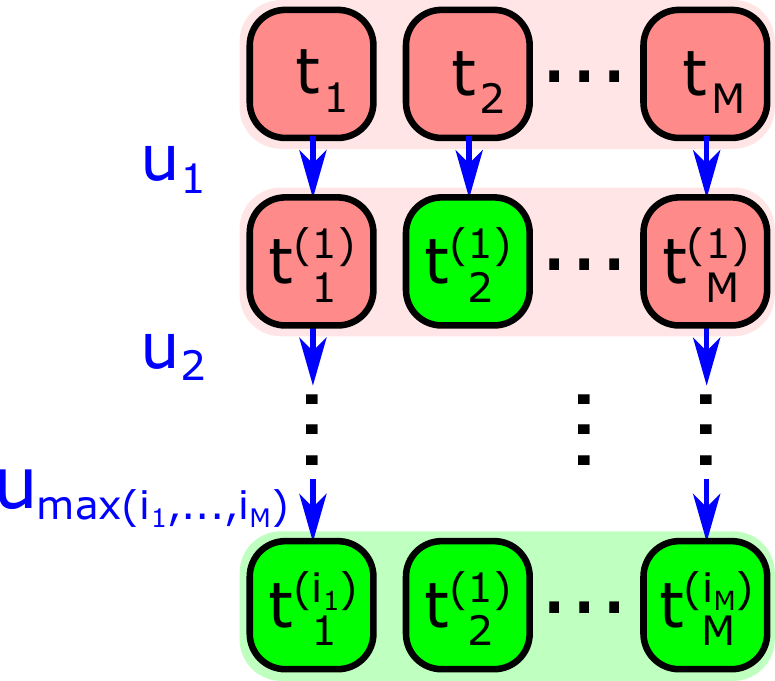}
\caption{Sequentially transforming the target gate set using elements from the mapping gate set (blue) until all elements of the latter are mapped to elements (highlighted in green). Once a target gate has been mapped to a native gate, the gate is not included in subsequent transformations. In this figure, $i_2 = 1$.}
\label{fig:seq}
\end{figure}
 
Another way of tackling the GSC problem is by what we shall refer to as \textit{simultaneous conversion}. In this multi-step approach, at every step, a mapping gate $m$ is applied to all elements in the set $T$ simultaneously. This approach is described by the GSC Algorithm, detailed below, which takes  $(T,N,M)$ as inputs and generates a sequence of mapping gates to transform all elements in $T$ to elements in $N$. The algorithm is illustrated in Fig.~\ref{fig:seq}. The protocol for choosing mapping gates $m$ is described by the conditional probability distribution $\pi^{(k)}$, e.g. mapping gates are chosen uniformly at random. This notation foreshadows the later uses of RL to tackle instances of the GSC problem, where $\pi^{(k)}$ will refer to the policy of an RL agent. In Algorithm \ref{algo:GSC}, when a mapping gate $m$ is said to be applied to a set $T$, the gate $m$ is applied to every element $t \in T$, i.e. $m^\dagger t m$. An element $t \in T$ is said to be mapped to an element in $N$, if a sequence of mapping gates $m$ was applied that transforms the target gate $t$ to a gate $t' \in N$. As soon as an element is mapped to an element in $N$, no further mapping gates are applied. The elements in $T$ are mapped in the order $t_{j_{1}},...,t_{j_{|T|}}$ induced by the chosen sequences of mapping gates. At step $i_{l}$, the element $t_{j_{l}}$ is mapped and we can retrieve that:
\begin{equation}
   n_{q_{l}} = \left(  \prod^{i_{l}}_{n=1} m_n \right)^\dagger t_{j_{l}} \left(  \prod^{i_{l}}_{n=1} m_n \right)
\end{equation}
From this, we can rewrite the product of target gates using the order retrieved by the GSC algorithm:


\begin{align}\label{eq:seq}
\begin{split}
\prod^{|T|}_{l=1} t_{j_l} &= \left(  \prod^{i_1}_{n=1} m_n \right) s_{q_1} \left(  \prod^{i_2}_{n=i_1+1} m_n \right) s_{q_2} \\
&\ldots \left( \prod^{i_{|T|}}_{n=i_{|T|-1}+1} m_n \right) s_{q_{|T|}} \left(\prod^{i_{|T|}}_{n=1} m_n \right) ^\dag \\
&= \left[ \prod^{|T|}_{l=1} \left( \prod^{i_l}_{n=i_{l-1}+1} m_n \right) s_{q_l} \right] \cdot  \left( \prod^{i_{|T|}}_{n=1} m_n \right)^\dag~,
\end{split}
\end{align}
where $i_0 = 0$. The main challenge is to find the shortest sequence mapping gates which, when applied sequentially to some elements of the native gate set, are equivalent to a product of all the target gates we are tasked to implement, since we are assuming the order of the target gates is not important (c.f. Eq.~\eqref{eq:timeevo}). A solution of this form will be referred to as simultaneous solution. 

\sloppy Given an individual solution, it is straightforward to obtain a corresponding simultaneous solution. Given that $(\prod^{i_{l}}_{n=1}  v_{j_{l_n}})^\dag  t_{j_l}(\prod v_{j_{l_n}}) = n_{k_l}$ for $l \in \{ 1, \ldots, |T| \}$, then we can set each product $\prod^{p_l}_{n=p_{l-1}+1} m_n = (\prod^{i_{l-1}}_{n=1} v_{j_{{l-1}_n}})^\dag (\prod^{i_l}_{n=1} v_{j_{l_n}})$ for $l \in \{ 2,\ldots,|T| \}$ and $\prod^{p_1}_{n=1} m_n = \prod^{i_1}_{n=1} v_{j_{1_n}}$, with the possible caveat of having to modify $M^u$ by adding any $v_{j_{l_n}} \notin M^u$. Note that the derived solution is indeed of the simultaneous conversion form since $p_l - p_{l-1} = i_{l-1} + i_l + 1 > 0$, so $p_1 < \ldots < p_{|T|}$. A schematic of this conversion is illustrated in Fig.~\ref{fig:sinseq}. 

We use the total number of mapping gates to compare solutions obtained by either a simultaneous or individual conversion. Further, for simultaneous solutions, we consider solution lengths with and without simplification of the product of mapping gates at the tail of the solution. More specifically, recall that a simultaneous solution has the form

\begin{equation}
\label{eq:seq_recall}
\prod^{|T|}_{l=1} t_{j_{l}} = \left[ \prod^{|T|}_{l=1} \left( \prod^{p_l}_{n=i_{p-1}+1} m_{n} \right) s_{k_{l}} \right] \cdot  \left( \prod^{p_{|T|}}_{n=1} m_{n} \right)^\dag~.
\end{equation}
\sloppy When the simultaneous form of the naive solution is considered, the tail $\left( \prod^{i_{|T|}}_{n=1} m_n \right)^\dag$ cancels, as can be seen in Sec.~\ref{sec:ind}, due to cancellations that occur between each neighboring sub-product of operators, e.g. $\prod^{p_l}_{n=p_{l-1}+1} m_n  \prod^{p_{l+1}}_{n=p_l+1} m_n = (\prod^{i_{l-1}}_{n=1} v_{j_{{l-1}_n}})^\dag (\prod^{i_l}_{n=1} v_{j_{l_n}}) (\prod^{i_l}_{n=1} v_{j_{l_n}})^\dag (\prod^{i_{l+1}}_{n=1} v_{j_{{l+1}_n}}) = (\prod^{i_{l-1}}_{n=1} v_{j_{{l-1}_n}})^\dag (\prod^{i_{l+1}}_{n=1} v_{j_{{l+1}_n}})$. Similarly, for each solution provided by the agent, we search for cancellations occurring between neighboring sub-products in the mapping gate sequence $\prod^{p_{|T|}}_{n=1} m_n$, marked by $p_1, \ldots, p_{|T|}$ in Eq.~\eqref{eq:seq_recall}. This circuit simplification is sketched in Fig.~\ref{fig:seq_cancellations}.


\begin{figure}[t!]
\centering
         \includegraphics[width=.4\linewidth]{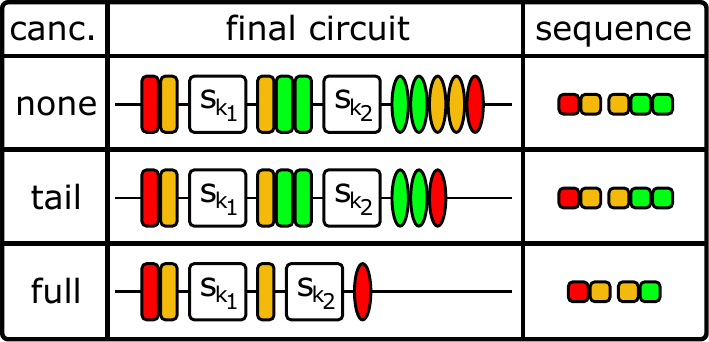}
        \caption{Illustration of the two different circuit simplifications one can perform on the circuit. The different mapping gates are illustrated by colored squares, with their conjugate being circles, whereas the sequence of mapping gates (for example, as suggested by an RL agent as we shall see in Sec.~\ref{sec:rl}). For tail cancellations, as the name implies, the simplifications only occur for the tail of the circuit, i.e. the complex conjugate of the mapping gate sequence. Thus, the effective sequence of operations doesn't change. For full cancellations, the entire circuit is simplified. 
        }
        \label{fig:seq_cancellations}
\end{figure}

\begin{algorithm}[htbp]
\caption{GSC Algorithm}

\textbf{Input} An instance of GSC $(T,N,M)$ (where $T$ and $N$ are w.l.o.g disjoint) is given. A step counter is set to $k=1$, a counter for removed elements is set to $l=1$ and the transformed target set is initialized to $T^{(k)}:= T$. A conditional probability distribution, a so-called policy, $\pi^{(k)}$ is chosen. From this distribution, given the transformed target set, a mapping gate $m$ is sampled. \\
\textbf{Output} A sequence of mapping gate operators $(m_1,\dots,m_K)$. 

\textbf{Procedure}
Repeat until $T^{(k)}=\emptyset$: 
    \begin{enumerate}
                \item A mapping gate $m_k$ is chosen according the policy $\pi^{(k)}$ and applied to the target set $T^{(k)}$, such that $T^{(k+1)}=\{ m_{k}^\dagger t^{(k)} m_{k}|\forall t^{(k)} \in T^{(k)}\}$ and then $k$ is incremented by one.
                \item For every element in $T$, if $t_j^{(k)}$ is equal to an element in $N$, this element is removed from the transformed target set $T^{(k)}:= T \setminus \{t_j^{(k)}\}$. The gate $t_j^{(k)}$ is the $l$-th removed element form the target set $T$, s.t. $t_j^{(k)}=t_{j_{l}}$. The mapping gate sequence $(m_{1},...,m_{i_{l}})$ where $i_{l}=k$ maps $t_j$ to the native $n_{q_{l}}=t_j^{(k)}$ gate $n_{q_{l}}$. If an element was removed $l$ is incremented by one. When all gates in $T$ are mapped to gates in $N$, i.e. $T^{(k)} = \emptyset$ the target set $T$ is resolved in $K=k$ steps.

    \end{enumerate}

\label{algo:GSC}
\end{algorithm}

\section{Computational complexity}\label{sec:complexity}

In this section, we sketch a proof that shows that the gate set conversion problem (GSC) is $\mathsf{NP}$-hard. The detailed proof of the following theorem can be found in the Appendix \ref{app:complexity}.
\begin{theorem}
The gate set conversion problem (GSC) is $\mathsf{NP}$-hard. 
\label{theorem1}
\end{theorem}

\begin{proof}
To prove this, we define a variant of the Hamiltonian path problem called Hamiltonian path with a starting vertex (HPS). In the following, the three problem classes HP, HPS, and GSC are defined in detail:

\begin{itemize}
\item \textsf{Hamiltonian path problem (HP)}:

Given an unweighted, undirected graph $G = (V, E)$, does $G$ have a Hamiltonian path, which is a sequence of edges that joins a sequence of vertices, with no vertex in the sequence repeated? This problem is already known to be $\mathsf{NP}$-hard \cite{karp_reducibility_1972}.

\item \textsf{Hamiltonian path problem with a starting vertex (HPS)}: 

Given an unweighted, undirected graph $G' = (V', E')$ and a node $s$, does $G'$ have a Hamiltonian path starting from vertex $s$?  

\item \textsf{Gate Set Conversion Decision Problem (GSDC)}:

Let $(T,N,M,k)$ be a tuple, where $T \subseteq T^u \coloneqq \{ X, Z \}^{\otimes Q}$ is the target set, $N \coloneqq \{ Z^{\otimes Q} \}$ is the native set, $M \subseteq M^u \coloneqq \{ I, H \}^{\otimes Q}$ the mapping gate set, $m_i \in M$ is a mapping gate with $i \in \{1,...,k\}$, and $k$ is an integer. Given an instance of GSCD $(T,N,M,k)$, can $T$ be resolved by a sequence of mapping gates $(m_1,...,m_{k})$ of length $k = |T|-1$?

\end{itemize}

We first show that there exists a polynomial time reduction from the Hamiltonian path problem (HP) to HPS. Then, we prove that there exists a polynomial time reduction from the HPS to GSCD to prove that GSC is $\mathsf{NP}$-hard.
\end{proof}
From Theorem~\ref{theorem1}, we can deduce the following corollary.
\begin{corollary}
The gate set conversion decision problem (GSCD) is $\mathsf{NP}$-complete.
\label{cor:GSCD}
\end{corollary}
\begin{proof}
    Since the GSCD is already proven to be $\mathsf{NP}$-hard, it remains to show that GSC is in $\mathsf{NP}$. The Algorithm~\ref{algo:GSCD} in Appendix~\ref{app:complexity} takes an instance $(T,N,M,K)$ as input and has a time complexity in the order of $O(K|T|)$. Thus, it allows for an efficient check of whether a given sequence resolves T showing that GSC is in $\mathsf{NP}$. 
\end{proof}

Given the hardness of the problem, we must resolve to approximate and heuristic algorithms to solve instances of GSCD.

\newpage

\section{Reinforcement Learning}
\label{sec:rl}


In reinforcement learning (RL) learning, the goal of a so-called learning agent is to adapt its behavior to maximize a given figure of merit \cite{bookSuttonBarto,bookSuttonBarto2nd}. The interaction between an RL agent and its environment can be mathematically described by a Markov Decision Process (MDP) \cite{Bellman1957}. An MDP is a 5-tuple $(S,A,s_0,R,P)$, where $S$ is the state space, i.e., the set of possible environmental states, $A$ the action space, i.e., the set of possible actions the agent can take, $s_0$ a starting state, $R: S\times A\rightarrow \mathbb{R}$ a reward function and $\mathrm{P}: S \times A \times S \rightarrow [0,1]$ is a transition function, i.e., the function that specifies the probability of transitioning to state $s'$, if in state $s$ the action $a$ was performed. At each time step $t$, the agent takes an action $a \in A$ and receives information about the environment in the form of a state $s\in S$ and a reward $r\in \mathbb{R}$. An \textit{episode} comprises all interactions between an agent and its environment until a termination condition is fulfilled. A standard figure of merit in such a scenario is the expected return: 
 \begin{equation}
    G_t= \sum_{k=t+1}^{\infty}\gamma^{k-t-1}r_k,
\end{equation}
where $r_k$ is the reward obtained at the k-th time step in the episode and $\gamma \in [0,1)$ is a discount factor that weights the contribution of future rewards. The reward $r_k$ is chosen to be zero for all $k$ after the termination of the episode. 
Assuming this figure of merit, each state and action pair $(s,a)$ can be assigned an action-value that quantifies the expected return starting from a state $s$ in step $t$ taking action $a$ and subsequently following \textit{policy} $\pi$: 
\begin{equation}
    q_\pi(s,a)=\mathbb{E}_{\pi}\left[ G_t|s, a \right]
\end{equation}
The behavior of a learning agent maximizing such a figure of merit is described by a conditional probability distribution called policy $\pi(a|s)$. The goal is to find an \textit{optimal} policy, i.e., a policy with a greater or equal expected return compared to all other policies for all states. The optimal policy can be derived from the optimal action-value function $q_*$. The Bellman optimality equation can be derived from the recursive relationship between the value of the current state and the next state:
\begin{equation}
   q_*(s,a)= \mathbb{E}\left[r_{t+1}+\max_{a'}q_*(s_{t+1},a')|s, a \right] 
\end{equation}
The solution of the Bellman optimality equation is an optimal policy. Instead of solving this equation analytically, in value-based RL, the goal is to derive the optimal action-value function from learned values estimated using data samples.
A well-known example of a value-based RL algorithm is $Q$-learning \cite{Watkins1989}, where each state-action pair $(s,a)$ is assigned a so-called $Q$-value $Q(s,a)$, which is updated to approximate $q_*$. Starting from an initial guess for all values $Q(s,a)$, the values are updated for each state-action pair $(s,a)$ while the agent interacts with the environment according to the following update rule:   
\begin{equation}
    Q(s,a)\leftarrow Q(s,a)+\alpha \left(r+\gamma \max_{a'} Q(s',a')-Q(s,a)\right),
\end{equation}
where $\alpha$ is the learning rate and $s'$ is the next encountered state after taking action $a$ in state $s$. 
The data for updates is sampled from the agent's policy. Thus, to guarantee learning the policy derived from the $Q$-values needs to be sufficiently explorative. 
A common choice is the $\epsilon$-greedy policy that, given the right parameters, guarantees exploration in the beginning and exploitation in the later stages of training: 
\begin{equation}
 \pi(a|s)= \begin{cases}
1-\epsilon_t \quad \text{for}\quad a=\argmax_{a'}Q(s,a')  \\
\epsilon_t \quad \text{otherwise}, \\
\end{cases}  
\end{equation}
where the parameter $\epsilon_t$ balances exploration and exploitation and is adapted over time.  

All $Q$-values can be stored in a table where the columns represent all actions and the rows represent all states. However, when the state space is large, storing the values in a table and updating them individually becomes infeasible. Instead, the entire action-value function can be approximated. In the following section, we describe how $Q$-learning can be extended to large state spaces using neural networks (NNs) as function approximators.

\subsection{Double Deep Q-learning}
In our work, we will utilize the method Double Deep Q-learning (DDQN). We made this choice as it has been successfully applied in other physics-inspired environments for example to optimize ansatzes for variational quantum circuits \cite{Ostaszewski2021} and in the future may benefit from quantum enhancements \cite{Jerbi2021}. 
DDQN is based on its predecessor Deep Q-learning (DQN), which is based on two essential methods for training neural networks (NN) in RL tasks.  
First, experience replay, a method to turn the sequential reinforcement learning data into the independently and identically distributed data required for NN training. In experience replay, the NN is trained with batches of experiences consisting of single-episode updates that are randomly sampled from a memory. 
Further, the NN training is stabilized by employing two NNs, a policy network, that is continuously updated, and a target network that is an earlier copy of the policy network. The policy network is used to estimate the current value, while the target network is used to provide a stable target value $Y$:
\begin{equation}
Y_\text{DQN}= r+\gamma \max_{a'} Q_\text{target}(s',a')
\end{equation}
In DQN, the policy network network is used to estimate the action values, which can lead to an overestimation bias resulting in unstable learning and a suboptimal policy. This is due to the maximization step over the action values in the term $\max_{a'} Q(s', a')$.
This issue is overcome in DDQN where the Q-function estimation is decoupled from the action selection. The target network is used for action-value selection and the policy network for action selection, each functioning as an independent estimator to reduce the maximization bias. 
\begin{equation}
  Y_\text{DDQN}= r+\gamma Q_\text{target}(s', \argmax_{a'} Q_\text{policy}(s',a')). 
\end{equation}
This target value will be approximated using a chosen loss function.


\subsection{Gate set conversion as a reinforcement learning problem}\label{sec:rlgsc}

Given an instance $(T,N,M)$ of the GSC, the state space of the corresponding MDP is the power set of the set of all Pauli strings $\mathbb{P}^Q$. The target gate set $T$ is the starting state $s_0=\{t_0,...,t_{|T|}\}$ of the environment. Then, in correspondence with the Algorithm \ref{algo:GSC}, the goal is to transform the state $s_0$ by applying a sequence of mapping gates $m$ until all elements of the set are mapped to elements in $N$ and removed from the state $s$. Thus, each state $s$ of the environment is a set of Pauli operators. Each action corresponds to the application of one mapping gate $m \in M$ to all elements of the current state $s$: 
\begin{equation}
\label{eq:action}
\begin{split}   
    a:\, S& \times M \rightarrow S\\
      (s&,m)\mapsto s'=\{t'|t \in s \text{ such that } t'=m^{\dagger}t m \}
\end{split}
\end{equation}
The transition function describes the transition from the current state $s$ to the next state $s'$: 
\begin{equation}
\label{eq:transition}
\begin{split}   
    f:\, S& \times A \rightarrow S\\
      (s&,a)\mapsto s'=\{t'|t \in s \text{ such that } t'=m^{\dagger}t m \text{ and } m^{\dagger}t m \notin N \}
\end{split}
\end{equation}

In this problem formulation, the goal of transforming all elements in $T$ to elements in $N$ can be simplified to transforming the state $s$ into the empty set. This goal can be translated into a binary reward function for the RL task. 

\begin{equation}
\label{eq:reward_dense}
\begin{split}   
    f:\, S \rightarrow \mathbb{R}\\
      (s)\mapsto \begin{cases}
  1 &\text{if } {s=\emptyset} \\
   0 &\text{otherwise} 
\end{cases}
\end{split}
\end{equation}

To facilitate learning in larger state spaces, we amend the binary reward with two additional terms for a denser reward landscape to increase sample efficiency. The resulting reward function is described by: 

\begin{equation}
\label{eq:reward}
\begin{split}   
    f:\, S& \times S \rightarrow \mathbb{R}\\
      (s&,s')\mapsto \begin{cases}
   d\cdot D+|s|-|s'| &\text{if } |s|>|s'| \\
   d\cdot D-C &\text{otherwise.} 
\end{cases}
\end{split}
\end{equation}
A common reward-shaping strategy is to add a constant negative reward for each time step. We scale this negative reward using the hyperparameter $C$.
The second hyperparameter $D$ is introduced to scale an additional reward. This reward is proportional to the difference $d=\sigma(s)-\sigma(s')$ of the \textit{similarity} of the current state and the next state to the native gate set.  
The distance $\sigma(s)$ of the state $s$ quantifies the similarity of the goal state to the native gate set and also constitutes a hyperparameter. To define the distance, we introduce a notion of an overlap between sets. Each environmental state can be described by a set of Pauli strings $\{P_{i}\}_{i=1}^l$ with $l\leq |T|$. First, we define the overlap of two Pauli strings $o(P_i,P_j)=Q-w(P_iP_j)$, where $w$ is the weight, which corresponds to the number of non-identity terms in the product of Pauli strings \cite{gottesman1997stabilizer}. Next, we can calculate the largest overlap between a single Pauli string and the native set $N=\{P_j\}_{j=1}^{|N|}$ of size $|N|$ using $o_\text{max}^N(P_i)=\max_{P_j \in N} o(P_i,P_j)$. Now, the overlap between the entire state $s$ and the native set is given by the sum over all largest overlaps $\rho^N(s)=\sum_{i=1}^{l}o^N_\text{max}((P_i)_{i=1}^Q)$. For example, if the current state is $s=\{XXIIYZ,IIIZXI\}$ then, with respect to a native set $N=\{IIIZZI,IIIIZZ\}$, the largest overlap for the Pauli string $XXIIYZ$ in $s$ is $o_{\max}^{N}(XXIIYZ)= 3$ and the largest overlap for the other Pauli string in $s$ is $o_{\max}^{N}(IIIZXI)=5$. This leads to an overlap between $s$ and the native gate set $N$ of $\sigma^{N}(s)=8$. The corresponding reward function that uses this additional shaped reward is used in the experiments discussed in the following section.

\section{Results}\label{sec:results}

\subsection{Learning performance}\label{sec:perform}
In this section, we present and analyze the numerical results of the learning performance under a varying target set size $|T|$ and qubit number $Q$. In the first set of experiments, the number of qubits is fixed to $Q=4$, while the target set size is chosen from $|T|\in\{8,12,16\}$. The learning performance in terms of the average mapping gate count $\overline{A_g}$ of the simultaneous solution during training is shown in Figure \ref{fig:training_target}. In a second set of experiments, the size of the target set is fixed to $|T|=8$ and the qubit number is chosen from $Q\in\{4,5,7\}$. The corresponding average mapping gate count $\overline{A_g}$ during training is shown in Figure \ref{fig:training_qubit}. The average is taken over 50 agents learning to solve the same instance of a GSC problem $(T,N,M)$. The error for the average count during training is estimated by the corresponding standard deviation. In both learning performance figures, a line in the same color as the learning performance indicates the length of the corresponding naive individual solution $N_\text{ind}$. This shows that the average learned performance at the end of the training lies below the naive individual solution. The state space $S$ of the RL task grows exponentially with the number of qubits. Given a qubit number, the size of the state space is determined by the number of combinations of Pauli strings of length smaller or equal to $|T|$. Thus, the state space grows polynomial with the size of the target set and the shape of the polynomial depends on the qubit number. Additionally, with each added element in the target set, the number of elements to be removed grows, but the reward density increases at the same time, alleviating part of the complexity of the learning problem. Thus, the agent's performance scales more favorably with increasing target set size. The hyperparameters chosen for the experiments are detailed in Appendix \ref{app:resultsrl}.  

\vspace{1.5cm}


 



\begin{figure}[htbp]
  \begin{minipage}{0.45\textwidth} 
    \includegraphics[width=\linewidth]{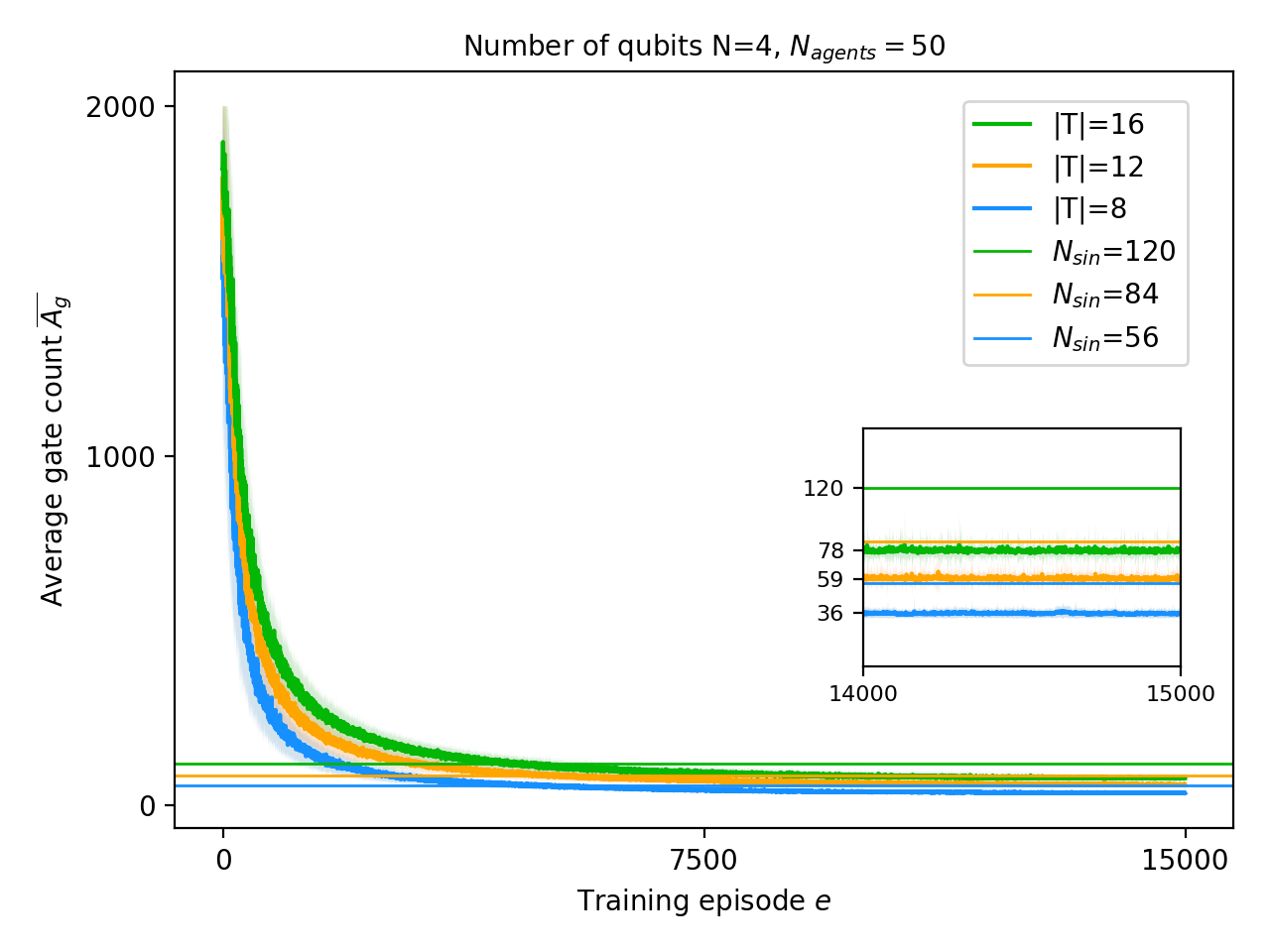}
    \caption{The agent's performance in terms of average gate count $\overline{A_g}$ for 50 agents during training at each episode for a single instance of GSC with 4-qubit gates and target gate set sizes $|T|\in \{8,12,16\}$. The same color line indicates the gate count of the corresponding naive individual solution $N_\text{ind}$. The error, indicated by the shaded area, is given by the standard deviation with a cutoff at the maximum gate count of 2000 and the lowest gate count obtained in all runs. }
    \label{fig:training_target}
  \end{minipage}%
  \hfill
    \begin{minipage}{0.45\textwidth}
    \includegraphics[width=\linewidth]{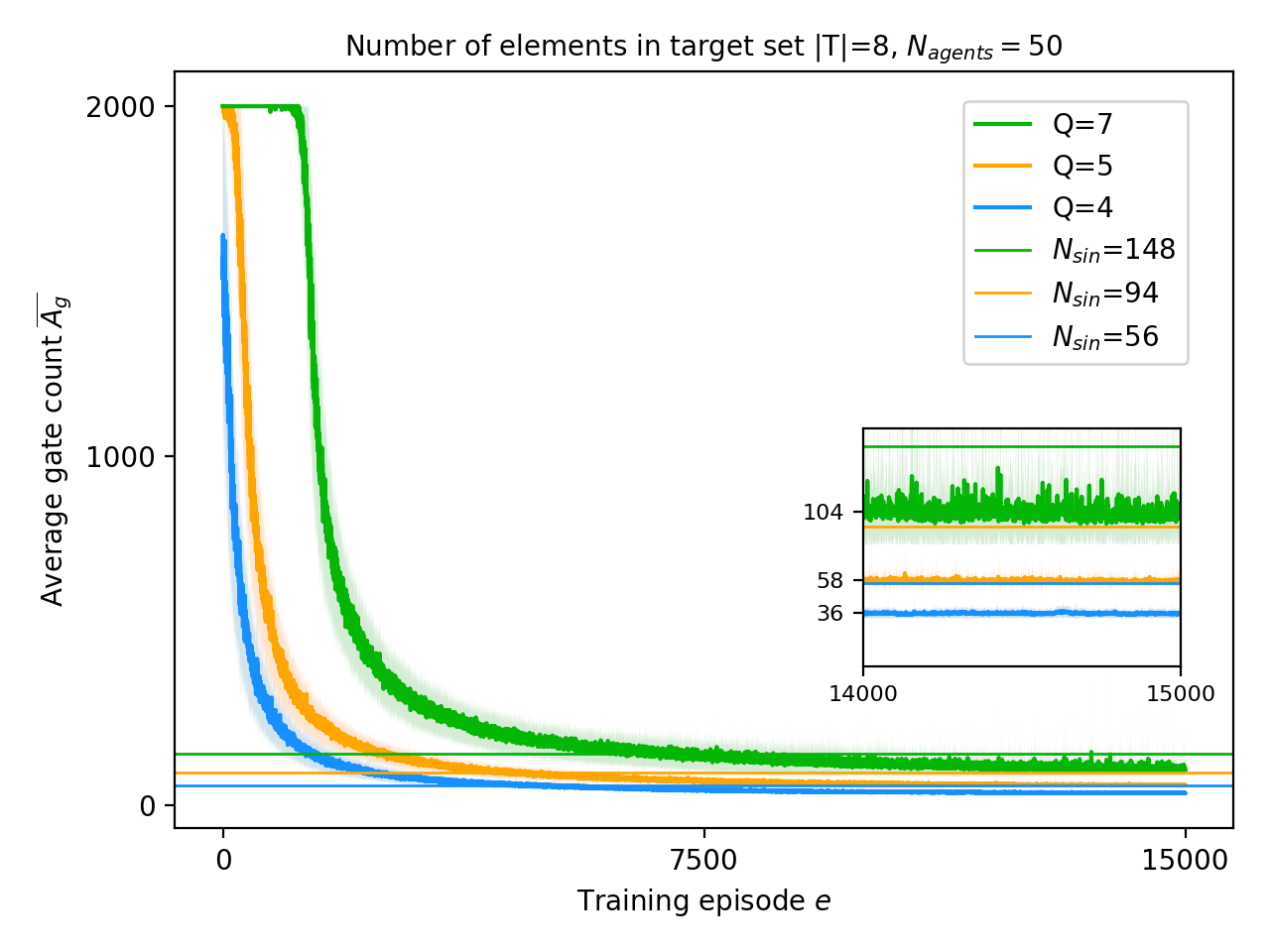}
    \caption{The agent's performance in terms of average gate count $\overline{A_g}$ for 50 agents during training at each episode for a single instance of GSC with target set size $|T_0|=8$ and qubit number $Q\in\{4,5,7\}$. The same color line indicates the gate count of the corresponding naive individual solution $N_\text{ind}$. The error, indicated by the shaded area, is given by the standard deviation with a cutoff at the maximum gate count of 2000 and the lowest gate count obtained in all runs.}
    \label{fig:training_qubit}
  \end{minipage}

  \end{figure}

\subsection{Shortest solution}\label{sec:shortest}

In solving a GSC problem instance, the aim is to find the shortest sequence of mapping gates that transforms all elements in the target set $T$ to elements in the native set $N$. Thus, in this section, we report the shortest obtained solutions for experiments with 4-qubit operators with target set sizes $|T|=8$, $|T|=12$, and $|T|=16$. In Table~\ref{tab:DDQN-all-t}, we provide the results for 3 agents trained on 5 different target gate sets. The shortest obtained solutions for experiments with operators on 4 to 7 qubits with a target set size of $|T|=8$ are shown in Table~\ref{tab:DDQN-all-q}. For each target set $T$ and agent $i\in\{1,2,3\}$, the lowest mapping gate count $A_{g_{i}}^{f}$, as well as the lowest mapping gate count taking into account full cancellations $A_{g_{i}}^{c}$, as defined in Sec.~\ref{sec:ind} are shown. From these results, we can see that the solutions obtained by the RL agents, only slightly reduce under full cancellations compared to the learned solutions. This is clearly not the case for the naive solution with and without these cancellations, as the simultaneous form is almost twice as long as the individual solution. Nonetheless, the agents find solutions that are significantly shorter than the naive solution in all experiments. 


\begin{table}[htbp] 
    \caption{Reinforcement learning results for 3 separate agents ($Ag_{1-3}$) tested on 5 different sets of target operators (indexed by the column $T$) on $4$ qubits with target set sizes $|T|=8$, $|T|=12$, and $|T|=16$, where $Ag^f_i$ and $Ag^c_i$ are the shortest solutions found by the $i^{\mathrm{th}}$ agent without and with full cancellations respectively. Both values are denoted in percent of the naive individual solution $N_{sim}$. For comparison, the number of mapping gates in the naive simultaneous ($N_{sim}$) and individual ($N_{ind}$) solutions are provided, which are averaged over 100 random orderings of the target gate set. Note that for the naive individual solution, the order of the operators in the target set is irrelevant.}\label{tab:DDQN-all-t}
    \begin{minipage}{\textwidth}
    \centering
        \subcaption{The lowest mapping gate countfor $|T|=8$.} 
        \begin{filecontents*}{TablesNew/table_4q_8t_table.csv}
seed;naivelen;naivelenOG;actionlena;actionlenb;actionlenc;actionlenatail;actionlenbtail;actionlenctail;
0;105.36;56;57\%;61\%;61\%;57\%;61\%;61\%;
1;138.56;74;51\%;49\%;51\%;49\%;49\%;51\%;
2;134.66;72;44\%;44\%;44\%;44\%;44\%;44\%;
3;134.7;72;53\%;50\%;50\%;50\%;50\%;50\%;
4;161.4;86;47\%;51\%;51\%;47\%;51\%;47\%;
\end{filecontents*}
\pgfplotstabletypeset[
        col sep=semicolon,
        string type,
        columns={seed, naivelen, naivelenOG, actionlena, actionlenb, actionlenc,  actionlenatail, actionlenbtail, actionlenctail},
        columns/seed/.style={column name=\textbf{T}, column type={|c|} },
        columns/naivelen/.style={column name=$\mathbf{N_{sim}}$},
        columns/actionlena/.style={column name=$\mathbf{Ag^f_1}$},
        columns/actionlenb/.style={column name=$\mathbf{Ag^f_2}$},
        columns/actionlenc/.style={column name=$\mathbf{Ag^f_3}$, column type={|c|} },
        columns/naivelenOG/.style={column name=$\mathbf{N_{ind}}$,column type={|c|}},
        columns/actionlenatail/.style={column name=$\mathbf{Ag^c_1}$},
        columns/actionlenbtail/.style={column name=$\mathbf{Ag^c_2}$},
        columns/actionlenctail/.style={column name=$\mathbf{Ag^c_3}$},
        every head row/.style={before row={\hline},after row=\hline\hline},
        every last row/.style={after row={\hline}},
        every nth row={1}{before row=\hline},
        column type/.add={|}{},
        every last column/.style={column type/.add={}{|}},
        ]{TablesNew/table_4q_8t_table.csv}

    \end{minipage}\\[3ex]

    \begin{minipage}{\textwidth}
    \centering
        \subcaption{The lowest mapping gate countfor $|T|=12$.}
             \begin{filecontents*}{TablesNew/table_4q_12t_table.csv}
seed;naivelen;naivelenOG;actionlena;actionlenb;actionlenc;actionlenatail;actionlenbtail;actionlenctail;
0;161.46;84;67\%;67\%;69\%;67\%;67\%;69\%;
1;184.02;96;48\%;50\%;50\%;46\%;50\%;48\%;
2;214.6;112;46\%;43\%;45\%;45\%;41\%;45\%;
3;218.4;114;46\%;49\%;49\%;46\%;47\%;47\%;
4;225.84;118;51\%;47\%;47\%;51\%;46\%;47\%;
\end{filecontents*}
\pgfplotstabletypeset[
        col sep=semicolon,
        string type,
        columns={seed, naivelen, naivelenOG, actionlena, actionlenb, actionlenc,  actionlenatail, actionlenbtail, actionlenctail},
        columns/seed/.style={column name=\textbf{T}, column type={|c|} },
        columns/naivelen/.style={column name=$\mathbf{N_{sim}}$},
        columns/actionlena/.style={column name=$\mathbf{Ag^f_1}$},
        columns/actionlenb/.style={column name=$\mathbf{Ag^f_2}$},
        columns/actionlenc/.style={column name=$\mathbf{Ag^f_3}$, column type={|c|} },
        columns/naivelenOG/.style={column name=$\mathbf{N_{ind}}$,column type={|c|}},
        columns/actionlenatail/.style={column name=$\mathbf{Ag^c_1}$},
        columns/actionlenbtail/.style={column name=$\mathbf{Ag^c_2}$},
        columns/actionlenctail/.style={column name=$\mathbf{Ag^c_3}$},
        every head row/.style={before row={\hline},after row=\hline\hline},
        every last row/.style={after row={\hline}},
        every nth row={1}{before row=\hline},
        column type/.add={|}{},
        every last column/.style={column type/.add={}{|}},
        ]{TablesNew/table_4q_12t_table.csv}

    \end{minipage}\\[3ex]

        \begin{minipage}{\textwidth}
        \centering
        \subcaption{The lowest mapping gate countfor $|T|=16$.}
           \begin{filecontents*}{TablesNew/table_4q_16t_table.csv}
seed;naivelen;naivelenOG;actionlena;actionlenb;actionlenc;actionlenatail;actionlenbtail;actionlenctail;
0;232.64;120;62\%;65\%;63\%;60\%;63\%;60\%;
1;228.84;118;44\%;44\%;44\%;44\%;41\%;41\%;
2;283.04;146;42\%;45\%;49\%;42\%;45\%;49\%;
3;334.22;172;40\%;40\%;38\%;40\%;40\%;38\%;
4;306.32;158;48\%;49\%;49\%;48\%;48\%;49\%;
\end{filecontents*}
\pgfplotstabletypeset[
        col sep=semicolon,
        string type,
        columns={seed, naivelen, naivelenOG, actionlena, actionlenb, actionlenc,  actionlenatail, actionlenbtail, actionlenctail},
        columns/seed/.style={column name=\textbf{T}, column type={|c|} },
        columns/naivelen/.style={column name=$\mathbf{N_{sim}}$},
        columns/actionlena/.style={column name=$\mathbf{Ag^f_1}$},
        columns/actionlenb/.style={column name=$\mathbf{Ag^f_2}$},
        columns/actionlenc/.style={column name=$\mathbf{Ag^f_3}$, column type={|c|} },
        columns/naivelenOG/.style={column name=$\mathbf{N_{ind}}$,column type={|c|}},
        columns/actionlenatail/.style={column name=$\mathbf{Ag^c_1}$},
        columns/actionlenbtail/.style={column name=$\mathbf{Ag^c_2}$},
        columns/actionlenctail/.style={column name=$\mathbf{Ag^c_3}$},
        every head row/.style={before row={\hline},after row=\hline\hline},
        every last row/.style={after row={\hline}},
        every nth row={1}{before row=\hline},
        column type/.add={|}{},
        every last column/.style={column type/.add={}{|}},
        ]{TablesNew/table_4q_16t_table.csv}

    \end{minipage}%
   
\end{table}


\begin{table}[htbp]
    \caption{Reinforcement learning results for 3 separate agents ($Ag_{1-3}$) tested on 5 different sets of target operators (indexed by the column $T$) with varying qubit and target set sizes, where $Ag^f_i$ and $Ag^c_i$ are the shortest solutions found by the $i^{\mathrm{th}}$ agent without and with full cancellations respectively. Both values are denoted in percent of the naive individual solution $N_{ind}$. For comparison, the number of mapping gates in the naive simultaneous ($N_{sim}$) and individual ($N_{ind}$) solutions are provided, which are averaged over 100 random orderings of the target gate set. Note that for the naive individual solution, the order of the operators in the target set is irrelevant.}
     \label{tab:DDQN-all-q}
       
        \begin{minipage}{\textwidth}
    \centering
          \subcaption{The lowest mapping gate countfor $Q=4$.}
        \pgfplotstabletypeset[
        col sep=semicolon,
        string type,
        columns={seed, naivelen, naivelenOG, actionlena, actionlenb, actionlenc,  actionlenatail, actionlenbtail, actionlenctail},
        columns/seed/.style={column name=\textbf{T}, column type={|c|} },
        columns/naivelen/.style={column name=$\mathbf{N_{sim}}$},
        columns/actionlena/.style={column name=$\mathbf{Ag^f_1}$},
        columns/actionlenb/.style={column name=$\mathbf{Ag^f_2}$},
        columns/actionlenc/.style={column name=$\mathbf{Ag^f_3}$, column type={|c|} },
        columns/naivelenOG/.style={column name=$\mathbf{N_{ind}}$,column type={|c|}},
        columns/actionlenatail/.style={column name=$\mathbf{Ag^c_1}$},
        columns/actionlenbtail/.style={column name=$\mathbf{Ag^c_2}$},
        columns/actionlenctail/.style={column name=$\mathbf{Ag^c_3}$},
        every head row/.style={before row={\hline},after row=\hline\hline},
        every last row/.style={after row={\hline}},
        every nth row={1}{before row=\hline},
        column type/.add={|}{},
        every last column/.style={column type/.add={}{|}},
        ]{TablesNew/table_4q_8t_table.csv}

    \end{minipage}\\[3ex]
        \begin{minipage}{\textwidth}
    \centering
          \subcaption{The lowest mapping gate countfor $Q=5$.}
        \begin{filecontents*}{TablesNew/table_5q_8t_table.csv}
seed;naivelen;naivelenOG;actionlena;actionlenb;actionlenc;actionlenatail;actionlenbtail;actionlenctail;
0;176.78;94;57\%;57\%;55\%;55\%;57\%;53\%;
1;190.82;102;49\%;49\%;49\%;49\%;49\%;49\%;
2;213.62;114;46\%;47\%;46\%;46\%;46\%;44\%;
3;179.9;96;56\%;58\%;58\%;52\%;54\%;56\%;
4;213.9;114;58\%;56\%;58\%;56\%;56\%;58\%;
\end{filecontents*}
\pgfplotstabletypeset[
        col sep=semicolon,
        string type,
        columns={seed, naivelen, naivelenOG, actionlena, actionlenb, actionlenc,  actionlenatail, actionlenbtail, actionlenctail},
        columns/seed/.style={column name=\textbf{T}, column type={|c|} },
        columns/naivelen/.style={column name=$\mathbf{N_{sim}}$},
        columns/actionlena/.style={column name=$\mathbf{Ag^f_1}$},
        columns/actionlenb/.style={column name=$\mathbf{Ag^f_2}$},
        columns/actionlenc/.style={column name=$\mathbf{Ag^f_3}$, column type={|c|} },
        columns/naivelenOG/.style={column name=$\mathbf{N_{ind}}$,column type={|c|}},
        columns/actionlenatail/.style={column name=$\mathbf{Ag^c_1}$},
        columns/actionlenbtail/.style={column name=$\mathbf{Ag^c_2}$},
        columns/actionlenctail/.style={column name=$\mathbf{Ag^c_3}$},
        every head row/.style={before row={\hline},after row=\hline\hline},
        every last row/.style={after row={\hline}},
        every nth row={1}{before row=\hline},
        column type/.add={|}{},
        every last column/.style={column type/.add={}{|}},
        ]{TablesNew/table_5q_8t_table.csv}

    \end{minipage}\\[3ex]
        \begin{minipage}{\textwidth}
    \centering
          \subcaption{The lowest mapping gate countfor $Q=6$.}
        \begin{filecontents*}{TablesNew/table_6q_8t_table.csv}
seed;naivelen;naivelenOG;actionlena;actionlenb;actionlenc;actionlenatail;actionlenbtail;actionlenctail;
0;225.68;120;60\%;60\%;63\%;60\%;58\%;62\%;
1;239.8;128;56\%;56\%;58\%;56\%;53\%;55\%;
2;299.76;160;50\%;55\%;54\%;50\%;55\%;54\%;
3;228.36;122;61\%;62\%;59\%;61\%;59\%;59\%;
4;258.78;138;54\%;58\%;54\%;52\%;58\%;52\%;
\end{filecontents*}
\pgfplotstabletypeset[
        col sep=semicolon,
        string type,
        columns={seed, naivelen, naivelenOG, actionlena, actionlenb, actionlenc,  actionlenatail, actionlenbtail, actionlenctail},
        columns/seed/.style={column name=\textbf{T}, column type={|c|} },
        columns/naivelen/.style={column name=$\mathbf{N_{sim}}$},
        columns/actionlena/.style={column name=$\mathbf{Ag^f_1}$},
        columns/actionlenb/.style={column name=$\mathbf{Ag^f_2}$},
        columns/actionlenc/.style={column name=$\mathbf{Ag^f_3}$, column type={|c|} },
        columns/naivelenOG/.style={column name=$\mathbf{N_{ind}}$,column type={|c|}},
        columns/actionlenatail/.style={column name=$\mathbf{Ag^c_1}$},
        columns/actionlenbtail/.style={column name=$\mathbf{Ag^c_2}$},
        columns/actionlenctail/.style={column name=$\mathbf{Ag^c_3}$},
        every head row/.style={before row={\hline},after row=\hline\hline},
        every last row/.style={after row={\hline}},
        every nth row={1}{before row=\hline},
        column type/.add={|}{},
        every last column/.style={column type/.add={}{|}},
        ]{TablesNew/table_6q_8t_table.csv}
    \end{minipage}\\[3ex]

    \begin{minipage}{\textwidth}
    \centering
        \subcaption{The lowest mapping gate countfor $Q=7$.}
        \begin{filecontents*}{TablesNew/table_7q_8t_table.csv}
seed;naivelen;naivelenOG;actionlena;actionlenb;actionlenc;actionlenatail;actionlenbtail;actionlenctail;
0;278.32;148;59\%;59\%;59\%;58\%;59\%;59\%;
1;299.98;160;56\%;59\%;56\%;55\%;59\%;56\%;
2;344.58;184;64\%;63\%;62\%;63\%;61\%;59\%;
3;311.12;166;58\%;57\%;59\%;58\%;57\%;58\%;
4;311.62;166;58\%;58\%;55\%;58\%;58\%;53\%;
\end{filecontents*}
\pgfplotstabletypeset[
        col sep=semicolon,
        string type,
        columns={seed, naivelen, naivelenOG, actionlena, actionlenb, actionlenc,  actionlenatail, actionlenbtail, actionlenctail},
        columns/seed/.style={column name=\textbf{T}, column type={|c|} },
        columns/naivelen/.style={column name=$\mathbf{N_{sim}}$},
        columns/actionlena/.style={column name=$\mathbf{Ag^f_1}$},
        columns/actionlenb/.style={column name=$\mathbf{Ag^f_2}$},
        columns/actionlenc/.style={column name=$\mathbf{Ag^f_3}$, column type={|c|} },
        columns/naivelenOG/.style={column name=$\mathbf{N_{ind}}$,column type={|c|}},
        columns/actionlenatail/.style={column name=$\mathbf{Ag^c_1}$},
        columns/actionlenbtail/.style={column name=$\mathbf{Ag^c_2}$},
        columns/actionlenctail/.style={column name=$\mathbf{Ag^c_3}$},
        every head row/.style={before row={\hline},after row=\hline\hline},
        every last row/.style={after row={\hline}},
        every nth row={1}{before row=\hline},
        column type/.add={|}{},
        every last column/.style={column type/.add={}{|}},
        ]{TablesNew/table_7q_8t_table.csv}
        
    \end{minipage}

\end{table}

\subsection{Comparison}\label{sec:comparison}
In this section, we compare the performance of the DDQN RL agents with two standard methods. First, we will use the widely used optimization algorithm, simulated annealing (SA). The implementation details and more results can be found in Appendix~\ref{app:resultssa}. Second, we will compare to the planning method Monte Carlo Tree Search (MCTS), which is defined in detail in the Appendix~\ref{app:resultsmcts}. To be able to compare the methods a fair amount of resources should be considered for each. Here, we compare the number of evaluations performed for each method. We define the total number of RL evaluations $N_{RL}$ as the number of queries to the reward function, which is given by the number of steps per episode summed over all episodes. For the total number of evaluations in the MCTS algorithm, as described in detail in Appendix \ref{app:resultsmcts}, we sum over the maximal tree depth reached in each episode and add the length of the naive solution used to evaluate the solution. In SA, we define the total number of evaluations $N_{SA}$ as the total number of queries to its cost function per repetition summed over all repetitions, as discussed in detail in Sec.~\ref{app:resultssa}. We chose the total number of RL evaluations to be smaller than the total number of MCTS or SA evaluations, where the total number of RL evaluations is $N_{RL}\approx 3\cdot 10^5$, the total number of MCTS evaluations is $N_{MCTS}\approx 4\cdot 10^5$ and the total number of SA evaluations is $N_{SA}\approx 5\cdot 10^5$. 
For each of the methods, we take the shortest found solutions for the 4-qubit GSC instances found in a coarse-grained parameter sweep. The shortest found solutions for the 4-qubit gate sets are shown in Table \ref{tab:comparison_4q}. The MCTS approach yields better results than SA, while the DDQN agent outperforms both methods in all experiments.


\begin{table}[htbp] 
    \caption{The results for reinforcement learning (RL), simulated annealing (SA), and Monte Carlo tree search (MCTS) tested on different sets of target operators (indexed by the column $T$) on $4$ qubits with target set sizes $|T|=8$. The results for the RL are denoted as $RL^f$ and $RL^c$ without and with full cancellations respectively. The results for the SA are denoted as $SA^f$ and $SA^c$ without and with full cancellations respectively. The results for the MCTS are denoted as $MCTS^f$ and $MCTS^c$ without and with full cancellations respectively. All values are denoted in percent of the naive individual solution $N_{ind}$. For comparison, the number of mapping gates in the naive simultaneous ($N_{sim}$) and individual ($N_{ind}$) solutions are provided, which are averaged over 100 random orderings of the target gate set. Note that for the naive individual solution, the order of the operators in the target set is irrelevant.}\label{tab:comparison_4q}
    \begin{minipage}{\textwidth}
    \centering
         \begin{filecontents*}{TablesNew/table_compare.csv}
seed;naivelen;naivelenred;rl;mcts;sa;rlred;mctsred;sared;
0;105.36;56;64\%;82\%;96\%;61\%;75\%;82\%;
1;138.56;74;54\%;68\%;114\%;51\%;59\%;76\%;
2;134.66;72;44\%;64\%;111\%;44\%;58\%;89\%;
3;134.7;72;53\%;61\%;117\%;53\%;58\%;78\%;
4;161.4;86;53\%;60\%;102\%;51\%;60\%;74\%;
\end{filecontents*}
\pgfplotstabletypeset[
        col sep=semicolon,
        string type,
        columns={seed, naivelen, naivelenred, rl, mcts, sa, rlred, mctsred, sared},
        columns/seed/.style={column name=\textbf{T},  column type={|c|} },
        columns/naivelen/.style={column name=$\mathbf{N_{sim}}$},
        columns/rl/.style={column name=$\mathbf{RL^f}$},
        columns/mcts/.style={column name=$\mathbf{MCTS^f}$},
        columns/sa/.style={column name=$\mathbf{SA^f}$, column type={|c|} },
        columns/naivelenred/.style={column name=$\mathbf{N_{ind}}$},
        columns/rlred/.style={column name=$\mathbf{RL^c}$},
        columns/mctsred/.style={column name=$\mathbf{MCTS^c}$},
        columns/sared/.style={column name=$\mathbf{SA^c}$},
        %
        %
        every head row/.style={before row={\hline},after row=\hline\hline},
        every last row/.style={after row={\hline}},
        every nth row={1}{before row=\hline},
        column type/.add={|}{},
        every last column/.style={column type/.add={}{|}},
        ]{TablesNew/table_compare.csv}
    \end{minipage}\\[3ex]

\end{table}

\subsection{Generalization for GSC problem}\label{sec:converter}

In this section, we analyze the generalization capabilities of the DDQN for the GSC problem. 
In all previous experiments, only a single starting state is used for training. However, neural network-based RL methods have the capability to learn and generalize over the entire state space, which allows to encode the solutions for an entire family of GSC instances in a single policy. 

Here, we compare the performance of a single agent trained on a number of different starting states. During training, at the beginning of each episode, a starting state is chosen uniformly at random for a set of starting states $S_0$. The size of the starting state set for training the agent is chosen from $|S_0| \in \{1,50,100,1000\}$. All target sets $T$ in the starting state set $S_0$ have the same size $|T|=8$. In Table \ref{tab:converter_compare}, we compare the average mapping gate count for 50 agents averaged over the last 1000 training episodes to the naive individual mapping gate count averaged over all states in the corresponding set $S_0$. The error is given by the standard deviation. These results show that the learned average solution length $\mu(A_g)$ for all $|S_0|=\{1,50,100,1000\}$ is well below the average naive individual solution length $\mu(N_{ind})$. Further, we want to shed light on how the single agent trained on 1000 states performs after training on each state individually. In Table~\ref{tab:converter_Z_50}, we show on how many of the states the single agent trained on 1000 states performs well, relative to the average performance of agents trained on those states separately. Since training all 1000 agents from scratch is rather resource-intensive, we choose a subset of 50 agents to train on 50 different starting states instead. The results show that, after training, on roughly half of the states the agent trained on 1000 states performs similarly or equal to the 50 agents tasked to learn only a single state. To further analyze the single agent trained on 1000 starting states, we compare its obtained solution to the naive individual solution length of all 1000 states. Table~\ref{tab:converter_Z_1000} shows for how many of the 1000 states the performance of a single agent trained on 1000 states is below a given percentage of the individual solution length. We can see from the average mapping gate count achieved during training of the 50 agents trained on states separately that the improvements over the naive individual solution generally varies from $45\%-65\%$. If we compare this to the single agent trained on 1000 states, we can see that such an agent reaches an improvement over the naive individual solution by at least $65\%$ on $Z_{1000}=479\pm14$ states. Training 1000 agents separately would amount to a cost of $15*10^{6}$ evaluations in terms of the number of episodes. A single agent trained trained on a thousand states reduces the evaluation cost by a factor of a thousand. Taking into consideration the performance after training, by training a single agent on 1000 states, we can achieve a reduction of the number of evaluations by a factor of 500 over the agents trained separately on all states.
The corresponding training performance in terms of the mapping gate count, for each episode $e$, is shown in Figure \ref{fig:converter}. In this figure, we can see that even though the complexity of the task increases due to the increasing number of states in the starting state set $|S_0|=1$ to $|S_0|=1000$, the performance of the RL agent during training is almost identical, showing that a DDQN agent can indeed generalize over the state space, learning to solve an entire family of GSC problems $\{(T,N,M)\}_{T \in S_0,|T|<K}$, where $K$ is some integer.  

\begin{figure}[t!]
\centering
\includegraphics[width=.7\linewidth]{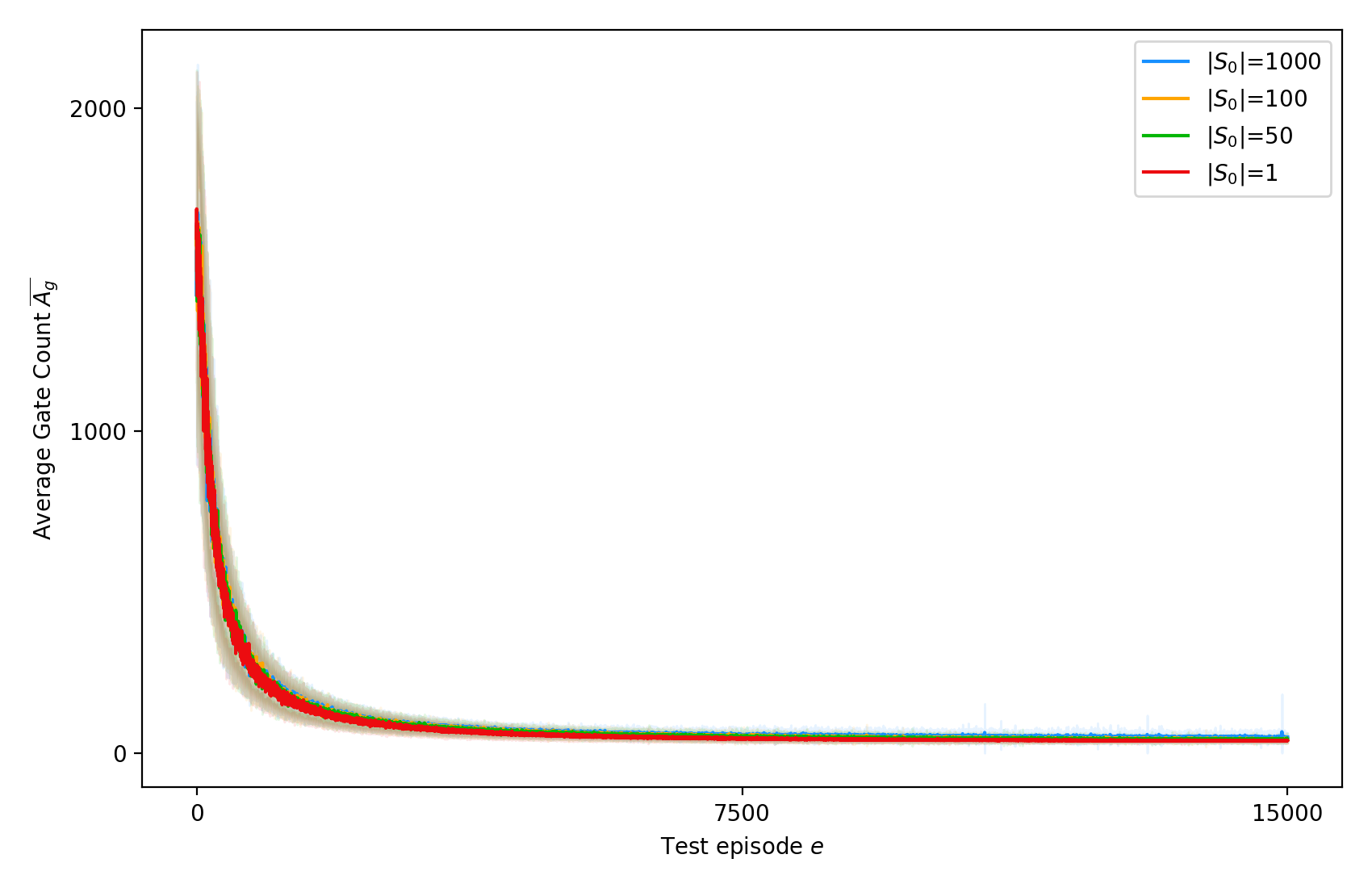}
\caption{The agent's performance in terms of the average mapping gate count $A_g$ for 40 agents during training at each episode on a 4-qubit and target set size $|T|=8$ problem instance with a varying number of starting states $|S_0|=\{1,50,100,1000\}$. The error, indicated by the shaded area, is given by the standard deviation with a cutoff at the maximum gate count of 2000 and the lowest gate count obtained in all runs. The agent's performance is similar even though the number of starting states differs. Thus, these curves indicate that training a single agent on many states is advantageous compared to training many agents separately.}
\label{fig:converter}
\end{figure}

\begin{table}[t!]
\begin{center}
\caption{The solution length of 50 agents averaged over the last 1000 episodes out of 15000 episodes is denoted as $\mu(A_g)$. The corresponding average naive individual solution length $\mu(N_\text{ind})$ is averaged over all elements in $S_0$. These results show that the learned solution length is below the naive individual solution length for all starting state set sizes $|S_0|$.}
\label{tab:converter_Z_50}
\begin{filecontents*}{TablesNew/converter_table_qubits=4_size=8.csv}
numstates;len;naivelen
1000;$48\pm11$;$71\pm12$
100;$42\pm8$;$70\pm11$
50;$42\pm7$;$69\pm12$
1;$39\pm7$;$72\pm0$
\end{filecontents*}
\pgfplotstabletypeset[
col sep=semicolon,
columns={numstates, len,  naivelen},
columns/numstates/.style={string type, column name=$\mathbf{|S_{0}|}$},
columns/len/.style={string type, column name=$\mu(A_g)$, column type={|c|} },
columns/naivelen/.style={string type, column name=$\mu(N_\text{ind})$},
every head row/.style={before row={\hline},after row=\hline\hline},
every last row/.style={after row={\hline}},
every nth row={1}{before row=\hline},
column type/.add={|}{},
every last column/.style={column type/.add={}{|}},
]{TablesNew/converter_table_qubits=4_size=8.csv} 
\end{center}

\end{table}

\begin{table}[t!]
\begin{center}
\caption{The average number of states $Z_{50}$ out of 50 states on which an agent trained on 1000 states for 15000 episodes achieves a given solution length relative to the average solution length of agents trained on the states separately. The relative solution length $A_{|S_0|=1000}/A_g$ is denoted as a percentage interval. The average of $Z_{50}$ is taken over 25 agents trained on $S_0$, while $\mu(A_g)$ is the solution length of the separately trained agents averaged over the last 1000 episodes of 15000 episodes of training on a single state. These results show a single agent trained on 1000 states outperforms agents trained separately on each state on $8\pm2$ of the 50 states and roughly matches the performance on around $20\pm3$ of the 50 states.}
\label{tab:converter_Z_1000}
\begin{filecontents*}{TablesNew/converter_table_50_states_compare.csv}
percent;numstates
$\leq100\%$;$8\pm2$
$100-115\%$;$20\pm3$
$115-130\%$;$15\pm3$
$\geq130\%$;$7\pm2$
\end{filecontents*}
\pgfplotstabletypeset[
col sep=semicolon,
string type,
columns={percent, numstates},
columns/numstates/.style={string type, column name=$\mathbf{Z_{50}}$},
columns/percent/.style={string type, column name=$A_{|S_0|=1000}/\mu(A_g)$, column type={|c|} },
every head row/.style={before row={\hline},after row=\hline\hline},
every last row/.style={after row={\hline}},
every nth row={1}{before row=\hline},
column type/.add={|}{},
every last column/.style={column type/.add={}{|}},
]{TablesNew/converter_table_50_states_compare.csv} 
\end{center}
\end{table}

\begin{table}[t!]
\begin{center}
\caption{The average number of states $Z_{1000}$ out of $|S_0|=1000$ states on which an agent, trained for 15000 episodes, achieves a given solution length relative to the naive individual solution. The average is taken over 25 agents trained on $S_0$. The relative solution lengths $A_{|S_0|=1000}/N_{ind}$ are denoted as percentage intervals. These results show that single agents trained on 1000 states achieve on around $50\%$ of the states a solution length which is reduced by at least $65\%$ compared to the naive individual solution. This indicates a reduction of the number of evaluations by a factor of 500 over the agents trained separately on all states.}
\label{tab:converter_compare}
\begin{filecontents*}{TablesNew/converter_table_all_naive_compare.csv}
percent;numstates
$\leq45\%$;$17\pm3$
$\leq55\%$;$156\pm9$
$\leq65\%$;$479\pm14$
$\leq75\%$;$774\pm11$
\end{filecontents*}
\pgfplotstabletypeset[
col sep=semicolon,
string type,
columns={percent, numstates},
columns/numstates/.style={string type, column name=$\mathbf{Z_{1000}}$},
columns/percent/.style={string type, column name=$A_{|S_0|=1000}/N_{ind}$, column type={|c|} },
every head row/.style={before row={\hline},after row=\hline\hline},
every last row/.style={after row={\hline}},
every nth row={1}{before row=\hline},
column type/.add={|}{},
every last column/.style={column type/.add={}{|}},
]{TablesNew/converter_table_all_naive_compare.csv} 
\end{center}

\end{table}

\newpage
\section{Discussion}
In this work, we discuss a common gate synthesis problem that arises due to hardware restrictions, where we are tasked to implement a product of target gates using products of available native and mapping gates. In order to benchmark our results with widely used techniques, we focused on a particular example of target, mapping, and native gates. After introducing the problem and discussing two ways of formulating a solution, we showed how even a relatively simple example of this problem is $\mathsf{NP}$-hard. We apply RL to tackle this problem and compare results with a standard approach of mapping sets of target gates to native gates used in previous literature. 

Our results show that RL not only surpasses the naive mapping strategy but also the two other methods tested, as in every instance tested, the RL agents were able to find circuits shorter than the ones returned by the naive, SA, or MCTS strategies (in some cases, significantly shorter). Whilst our results on 4-7 qubits hint at the efficacy of RL for this problem, our approach can be extended to larger systems using methods that increase sample efficiency. Furthermore, we believe that our work opens up a number of interesting avenues to pursue in the near future. Future work might investigate transfer learning applied to this gate synthesis problem, wherein one might consider training the agent on a smaller system size, and then use the knowledge gained on larger problems (more qubits, target gates, etc.), with the aim of speeding up the learning process. This latter approach would of course be of the most substantial value as scaling is always the number one problem, and we expect it to be hard but possible. Another research avenue one might follow is extending this work to Variational Quantum Eigensolver ansatzes and other practical applications, as we focused more on an abstract application of the methods discussed. 

\section*{Acknowledgements}
LMT acknowledges the support by the Austrian Science Fund (FWF) through the DK-ALM: W1259-N27 and SFB BeyondC F7102.
VD and ES acknowledge the support of SURF through the QC4QC project.
This work was supported by the Dutch Research Council (NWO/OCW), as part of the Quantum Software Consortium programme (project number 024.003.037). This work was also supported by the Dutch National Growth Fund (NGF), as part of the Quantum Delta NL programme.

\printbibliography 

@article{Moro2021,
	title = {Quantum compiling by deep reinforcement learning},
	volume = {4},
	rights = {2021 The Author(s)},
	issn = {2399-3650},
	url = {https://www.nature.com/articles/s42005-021-00684-3},
	doi = {10.1038/s42005-021-00684-3},
	pages = {1--8},
	number = {1},
	journaltitle = {Communications Physics},
	shortjournal = {Commun Phys},
	author = {Moro, Lorenzo and Paris, Matteo G. A. and Restelli, Marcello and Prati, Enrico},
	urldate = {2023-05-24},
	date = {2021-08-06},
	langid = {english},
	note = {Number: 1
Publisher: Nature Publishing Group},
}

@article{Zhang2020,
	title = {Topological Quantum Compiling with Reinforcement Learning},
	volume = {125},
	url = {https://link.aps.org/doi/10.1103/PhysRevLett.125.170501},
	doi = {10.1103/PhysRevLett.125.170501},
	pages = {170501},
	number = {17},
	journaltitle = {Physical Review Letters},
	shortjournal = {Phys. Rev. Lett.},
	author = {Zhang, Yuan-Hang and Zheng, Pei-Lin and Zhang, Yi and Deng, Dong-Ling},
	urldate = {2023-05-24},
	date = {2020-10-19},
	note = {Publisher: American Physical Society},
}

@article{zhao2022making,
    title = {Making Trotterization adaptive for NISQ devices and beyond},
    author = {Zhao, Hongzheng and Bukov, Marin and Heyl, Markus and Moessner, Roderich},
    journal = {arXiv:2209.12653},
    year = {2022},
    url = {https://arxiv.org/abs/2209.12653}
}

@article{tran2020destructive,
    title = {Destructive Error Interference in Product-Formula Lattice Simulation},
    author = {Tran, Minh C. and Chu, Su-Kuan and Su, Yuan and Childs, Andrew M. and Gorshkov, Alexey V.},
    journal = {Phys. Rev. Lett.},
    volume = {124},
    pages = {220502},
    year = {2020},
    url = {https://arxiv.org/abs/1912.11047}
}

@article{rendon2022improved,
    title = {Improved Error Scaling for Trotter Simulations through Extrapolation},
    author = {Rendon, Gumaro and Watkins, Jacob and Wiebe, Nathan},
    journal = {ArXiv:2212.14144},
    url = {https://arxiv.org/abs/2212.14144},
    year = {2022}
}

@article{mukhopadhyay2022synthesizing,
    title = {Synthesizing efficient circuits for Hamiltonian simulation},
    author = {Mukhopadhyay, Priyanka and Wiebe, Nathan and Zhang, Hong Tao},
    journal = {npj Quant. Inf.},
    volume = {9},
    issue = {31},
    year = {2023},
    url = {https://arxiv.org/abs/2209.03478}
}

@article{heyl2019quantum,
    title = {Quantum localization bounds Trotter errors in digital quantum simulation},
    author = {Heyl, Markus and Hauke, Philipp and Zoller, Peter},
    journal = {Sci. Adv.},
    volume = {5},
    pages = {eeau8342},
    year = {2019},
    url = {https://arxiv.org/abs/1806.11123}
}

@article{wan2022randomized,
    title = {A randomized quantum algorithm for statistical phase estimation},
    author = {Wan, Kianna and Berta, Mario and Campbell, Earl T.},
    journal = {Phys. Rev. Lett.},
    volume = {129},
    pages = {030503},
    year = {2022},
    url = {https://arxiv.org/abs/2110.12071}
}

@article{campbell2018random,
    title = {A random compiler for fast Hamiltonian simulation},
    author = {Earl Campbell},
    journal = {Phys. Rev. Lett.},
    volume = {123},
    pages = {070503},
    year = {2019},
    url = {https://arxiv.org/abs/1811.08017}
}

@article{hastings2014improving,
    title = {Improving Quantum Algorithms for Quantum Chemistry},
    author = {Hastings, M. B. and Wecker, D. and Bauer, B. and Troyer, M.},
    journal = {Quant. Inf. Comp.},
    volume = {15},
    issue = {1},
    year = {2015},
    url = {https://arxiv.org/abs/1403.1539}
}

@article{hagan2022composite,
    title = {Composite Quantum Simulations},
    author = {Hagan, Matthew and Wiebe, Nathan},
    journal = {ArXiv:2206.060409},
    url = {https://arxiv.org/abs/2206.06409},
    year = {2022}
}

@article{lloyd1996universal,
    title = {Universal Quantum Simulators},
    author = {Lloyd, Seth},
    journal = {Science},
    volume = {273},
    pages = {103},
    year = {1996},
    url = {https://www.science.org/doi/10.1126/science.273.5278.1073}
}

@article{suzuki1991general,
    title = {General Theory of Fractal Path Integrals with Applications to Many-Body Theories and Statistical Physics},
    journal = {J. Math. Phys. (N.Y.)},
    volume = {32},
    pages = {400},
    year = {1991},
    url = {https://pubs.aip.org/aip/jmp/article-abstract/32/2/400/229229/General-theory-of-fractal-path-integrals-with?redirectedFrom=fulltext}
}

@article{childs2019theory,
    title = {Theory of Trotter Error with Commutator Scaling},
    author = {Childs, Andrew M. and Su, Yuan and Tran, Minh C. and Wiebe, Nathan and Zhu, Shuchen},
    journal = {Phys. Rev. X},
    volume = {11},
    pages = {011020},
    year = {2021},
    url = {https://journals.aps.org/prx/abstract/10.1103/PhysRevX.11.011020}
}

@article{babbush2015chemical,
    title = {Chemical Basis of Trotter-Suzuki Errors in Quantum Chemistry Simulation},
    author = {Babbush, Ryan and McClean, Jarrod and Wecker, Dave and Aspuru-Guzik, Al{\' a}n and Wiebe, Nathan},
    journal = {Phys. Rev. A},
    volume = {91},
    pages = {022311},
    year = {2015},
    url = {https://arxiv.org/abs/1410.8159}
}

@article{kivlichan2019improved,
    title = {Improved Fault-Tolerant Quantum Simulation of Condensed-Phase Correlated Electrons via Trotterization},
    author = {Kivlichan, Ian D. and Gidney, Craig and Berry, Dominic W. and Wiebe, Nathan and McClean, Jarrod and Sun, Wei and Jiang, Zhang and Rubin, Nicholas and Fowler, Austin and Aspuru-Guzik, Al{\' a}n and Neven, Hartmut and Babbush, Ryan},
    journal = {Quantum},
    volume = {4},
    pages = {296},
    year = {2020},
    url = {https://quantum-journal.org/papers/q-2020-07-16-296/}
}

@article{low2017optimal,
    title = {Optimal Hamiltonian Simulation by Quantum Signal Processing},
    author = {Low, Guang Hao and Chuang, Isaac L.},
    journal = {Phys. Rev. Lett.},
    volume = {118},
    pages = {010501},
    year = {2017},
    url = {https://arxiv.org/abs/1606.02685}
}

@article{childs2012hamiltonian,
    title = {Hamiltonian Simulation Using Linear Combinations of Unitary Operations},
    author = {Childs, Andrew M. and Wiebe, Nathan},
    journal = {Quant. Inf. Comp.},
    volume = {12},
    pages = {901-924},
    year = {2012},
    url = {https://arxiv.org/abs/1202.5822}
}

@article{Jerbi2021,
archivePrefix = {arXiv},
arxivId = {1910.12760},
author = {Jerbi, Sofiene and Trenkwalder, Lea M. and {Poulsen Nautrup}, Hendrik and Briegel, Hans J. and Dunjko, Vedran},
doi = {10.1103/prxquantum.2.010328},
eprint = {1910.12760},
issn = {2691-3399},
journal = {PRX Quantum},
month = {9},
number = {1},
title = {{Quantum Enhancements for Deep Reinforcement Learning in Large Spaces}},
url = {http://arxiv.org/abs/1910.12760},
volume = {2},
year = {2021}
}

@article{shor1997,
author = {Shor, Peter W.},
title = {Polynomial-Time Algorithms for Prime Factorization and Discrete Logarithms on a Quantum Computer},
journal = {SIAM Journal on Computing},
volume = {26},
number = {5},
pages = {1484-1509},
year = {1997},
doi = {10.1137/S0097539795293172},
URL = {https://doi.org/10.1137/S0097539795293172},
eprint = {https://doi.org/10.1137/S0097539795293172}
}

@online{farhi2014,
      title={A Quantum Approximate Optimization Algorithm}, 
      author={Edward Farhi and Jeffrey Goldstone and Sam Gutmann},
      year={2014},
      eprint={1411.4028},
      archivePrefix={arXiv},
      primaryClass={quant-ph}
}

@article{maslov2008,
  author={D. {Maslov} and G. W. {Dueck} and D. M. {Miller} and C. {Negrevergne}},
  journal={IEEE Transactions on Computer-Aided Design of Integrated Circuits and Systems}, 
  title={Quantum Circuit Simplification and Level Compaction}, 
  year={2008},
  volume={27},
  number={3},
  pages={436-444},
  doi={10.1109/TCAD.2007.911334},
  ISSN={1937-4151},
  month={3},}

@article{fowler2011,
author = {Fowler, Austin G.},
title = {Constructing Arbitrary Steane Code Single Logical Qubit Fault-Tolerant Gates},
year = {2011},
issue_date = {September 2011},
publisher = {Rinton Press, Incorporated},
address = {Paramus, NJ},
volume = {11},
number = {9–10},
issn = {1533-7146},
journal = {Quantum Info. Comput.},
month = {9},
pages = {867–873},
numpages = {7}
}

@article{chen2013,
title = {Qcompiler: Quantum compilation with the CSD method},
journal = {Computer Physics Communications},
volume = {184},
number = {3},
pages = {853-865},
year = {2013},
issn = {0010-4655},
doi = {https://doi.org/10.1016/j.cpc.2012.10.019},
url = {https://www.sciencedirect.com/science/article/pii/S0010465512003621},
author = {Y.G. Chen and J.B. Wang}
}

@article{chong2017,
author={Chong, Frederic T.
and Franklin, Diana
and Martonosi, Margaret},
title={Programming languages and compiler design for realistic quantum hardware},
journal={Nature},
year={2017},
month={9},
day={01},
volume={549},
number={7671},
pages={180-187},
issn={1476-4687},
doi={10.1038/nature23459},
url={https://doi.org/10.1038/nature23459}
}

@article{ferrari2018,
author = {Ferrari, Davide and Amoretti, Michele},
title = {Efficient and effective quantum compiling for entanglement-based machine learning on IBM Q devices},
journal = {International Journal of Quantum Information},
volume = {16},
number = {08},
pages = {1840006},
year = {2018},
doi = {10.1142/S0219749918400063},
URL = { https://doi.org/10.1142/S0219749918400063},
eprint = {https://doi.org/10.1142/S0219749918400063}
}

@online{booth2018,
      title={Comparing and Integrating Constraint Programming and Temporal Planning for Quantum Circuit Compilation}, 
      author={Kyle E. C. Booth and Minh Do and J. Christopher Beck and Eleanor Rieffel and Davide Venturelli and Jeremy Frank},
      year={2018},
      eprint={1803.06775},
      archivePrefix={arXiv},
      primaryClass={quant-ph}
}

@article{venturelli2018,
	doi = {10.1088/2058-9565/aaa331},
	url = {https://doi.org/10.1088/2058-9565/aaa331},
	year = 2018,
	month = {2},
	publisher = {{IOP} Publishing},
	volume = {3},
	number = {2},
	pages = {025004},
	author = {Davide Venturelli and Minh Do and Eleanor Rieffel and Jeremy Frank},
	title = {Compiling quantum circuits to realistic hardware architectures using temporal planners},
	journal = {Quantum Science and Technology}
}

@article{cincio2018,
	doi = {10.1088/1367-2630/aae94a},
	url = {https://doi.org/10.1088/1367-2630/aae94a},
	year = 2018,
	month = {11},
	publisher = {{IOP} Publishing},
	volume = {20},
	number = {11},
	pages = {113022},
	author = {Lukasz Cincio and Yi{\u{g}}it Suba{\c{s}}{\i} and Andrew T Sornborger and Patrick J Coles},
	title = {Learning the quantum algorithm for state overlap},
	journal = {New Journal of Physics}
}

@article{nam2018,
author={Nam, Yunseong
and Ross, Neil J.
and Su, Yuan
and Childs, Andrew M.
and Maslov, Dmitri},
title={Automated optimization of large quantum circuits with continuous parameters},
journal={npj Quantum Information},
year={2018},
month={5},
day={10},
volume={4},
number={1},
pages={23},
issn={2056-6387},
doi={10.1038/s41534-018-0072-4},
url={https://doi.org/10.1038/s41534-018-0072-4}
}

@article{heyfron2018,
	doi = {10.1088/2058-9565/aad604},
	url = {https://doi.org/10.1088/2058-9565/aad604},
	year = 2018,
	month = {11},
	publisher = {{IOP} Publishing},
	volume = {4},
	number = {1},
	pages = {015004},
	author = {Luke E Heyfron and Earl T Campbell},
	title = {An efficient quantum compiler that reduces T count},
	journal = {Quantum Science and Technology}
}

@inproceedings{gokhale2019,
author = {Gokhale, Pranav and Ding, Yongshan and Propson, Thomas and Winkler, Christopher and Leung, Nelson and Shi, Yunong and Schuster, David I. and Hoffmann, Henry and Chong, Frederic T.},
title = {Partial Compilation of Variational Algorithms for Noisy Intermediate-Scale Quantum Machines},
year = {2019},
isbn = {9781450369381},
publisher = {Association for Computing Machinery},
url = {https://doi.org/10.1145/3352460.3358313},
doi = {10.1145/3352460.3358313},
booktitle = {Proceedings of the 52nd Annual IEEE/ACM International Symposium on Microarchitecture},
pages = {266–278},
numpages = {13},
series = {MICRO '52}
}

@article{khatri2019,
  doi = {10.22331/q-2019-05-13-140},
  url = {https://doi.org/10.22331/q-2019-05-13-140},
  title = {Quantum-assisted quantum compiling},
  author = {Khatri, Sumeet and LaRose, Ryan and Poremba, Alexander and Cincio, Lukasz and Sornborger, Andrew T. and Coles, Patrick J.},
  journal = {{Quantum}},
  issn = {2521-327X},
  publisher = {{Verein zur F{\"{o}}rderung des Open Access Publizierens in den Quantenwissenschaften}},
  volume = {3},
  pages = {140},
  month = {5},
  year = {2019}
}

@article{duncan2020,
  doi = {10.22331/q-2020-06-04-279},
  url = {https://doi.org/10.22331/q-2020-06-04-279},
  title = {Graph-theoretic {S}implification of {Q}uantum {C}ircuits with the {ZX}-calculus},
  author = {Duncan, Ross and Kissinger, Aleks and Perdrix, Simon and van de Wetering, John},
  journal = {{Quantum}},
  issn = {2521-327X},
  publisher = {{Verein zur F{\"{o}}rderung des Open Access Publizierens in den Quantenwissenschaften}},
  volume = {4},
  pages = {279},
  month = {6},
  year = {2020}
}

@article{camps2020,
  title = {Approximate quantum circuit synthesis using block encodings},
  author = {Camps, Daan and Van Beeumen, Roel},
  journal = {Phys. Rev. A},
  volume = {102},
  issue = {5},
  pages = {052411},
  numpages = {7},
  year = {2020},
  month = {11},
  publisher = {American Physical Society},
  doi = {10.1103/PhysRevA.102.052411},
  url = {https://link.aps.org/doi/10.1103/PhysRevA.102.052411}
}

@online{jones2020,
      title={Quantum compilation and circuit optimisation via energy dissipation}, 
      author={Tyson Jones and Simon C Benjamin},
      year={2020},
      eprint={1811.03147},
      archivePrefix={arXiv},
      primaryClass={quant-ph}
}

@article{fosel2018,
  title = {Reinforcement Learning with Neural Networks for Quantum Feedback},
  author = {F\"osel, Thomas and Tighineanu, Petru and Weiss, Talitha and Marquardt, Florian},
  journal = {Phys. Rev. X},
  volume = {8},
  issue = {3},
  pages = {031084},
  numpages = {15},
  year = {2018},
  month = {9},
  publisher = {American Physical Society},
  doi = {10.1103/PhysRevX.8.031084},
  url = {https://link.aps.org/doi/10.1103/PhysRevX.8.031084}
}

@article{nautrup2019,
  doi = {10.22331/q-2019-12-16-215},
  url = {https://doi.org/10.22331/q-2019-12-16-215},
  title = {Optimizing {Q}uantum {E}rror {C}orrection {C}odes with {R}einforcement {L}earning},
  author = {Nautrup, Hendrik Poulsen and Delfosse, Nicolas and Dunjko, Vedran and Briegel, Hans J. and Friis, Nicolai},
  journal = {{Quantum}},
  issn = {2521-327X},
  publisher = {{Verein zur F{\"{o}}rderung des Open Access Publizierens in den Quantenwissenschaften}},
  volume = {3},
  pages = {215},
  month = {12},
  year = {2019}
}

@article{colomer2020,
title = {Reinforcement learning for optimal error correction of toric codes},
journal = {Physics Letters A},
volume = {384},
number = {17},
pages = {126353},
year = {2020},
issn = {0375-9601},
doi = {https://doi.org/10.1016/j.physleta.2020.126353},
url = {https://www.sciencedirect.com/science/article/pii/S0375960120301638},
author = {Laia {Domingo Colomer} and Michalis Skotiniotis and Ramon Muñoz-Tapia}
}

@article{bukov2018,
  title = {Reinforcement Learning in Different Phases of Quantum Control},
  author = {Bukov, Marin and Day, Alexandre G. R. and Sels, Dries and Weinberg, Phillip and Polkovnikov, Anatoli and Mehta, Pankaj},
  journal = {Phys. Rev. X},
  volume = {8},
  issue = {3},
  pages = {031086},
  numpages = {15},
  year = {2018},
  month = {9},
  publisher = {American Physical Society},
  doi = {10.1103/PhysRevX.8.031086},
  url = {https://link.aps.org/doi/10.1103/PhysRevX.8.031086}
}

@online{alam2019,
      title={Quantum Logic Gate Synthesis as a Markov Decision Process}, 
      author={M. Sohaib Alam},
      year={2019},
      eprint={1912.12002},
      archivePrefix={arXiv},
      primaryClass={quant-ph}
}

@article{zhang2019,
author={Zhang, Xiao-Ming
and Wei, Zezhu
and Asad, Raza
and Yang, Xu-Chen
and Wang, Xin},
title={When does reinforcement learning stand out in quantum control? A comparative study on state preparation},
journal={npj Quantum Information},
year={2019},
month={10},
day={08},
volume={5},
number={1},
pages={85},
issn={2056-6387},
doi={10.1038/s41534-019-0201-8},
url={https://doi.org/10.1038/s41534-019-0201-8}
}

@article{mackeprang2020,
author={Mackeprang, Jelena
and Dasari, Durga B. Rao
and Wrachtrup, J{\"o}rg},
title={A reinforcement learning approach for quantum state engineering},
journal={Quantum Machine Intelligence},
year={2020},
month={5},
day={11},
volume={2},
number={1},
pages={5},
issn={2524-4914},
doi={10.1007/s42484-020-00016-8},
url={https://doi.org/10.1007/s42484-020-00016-8}
}

@online{yao2020,
      title={Reinforcement Learning for Many-Body Ground State Preparation based on Counter-Diabatic Driving}, 
      author={Jiahao Yao and Lin Lin and Marin Bukov},
      year={2020},
      eprint={2010.03655},
      archivePrefix={arXiv},
      primaryClass={quant-ph}
}

@online{mckiernan2019,
      title={Automated quantum programming via reinforcement learning for combinatorial optimization}, 
      author={Keri A. McKiernan and Erik Davis and M. Sohaib Alam and Chad Rigetti},
      year={2019},
      eprint={1908.08054},
      archivePrefix={arXiv},
      primaryClass={quant-ph}
}

@online{khairy2019,
      title={Reinforcement-Learning-Based Variational Quantum Circuits Optimization for Combinatorial Problems}, 
      author={Sami Khairy and Ruslan Shaydulin and Lukasz Cincio and Yuri Alexeev and Prasanna Balaprakash},
      year={2019},
      eprint={1911.04574},
      archivePrefix={arXiv},
      primaryClass={cs.LG}
}

@article{beloborodov2021,
	doi = {10.1088/2632-2153/abc328},
	url = {https://doi.org/10.1088/2632-2153/abc328},
	year = 2021,
	month = {1},
	publisher = {{IOP} Publishing},
	volume = {2},
	number = {2},
	pages = {025009},
	author = {Dmitrii Beloborodov and A E Ulanov and Jakob N Foerster and Shimon Whiteson and A I Lvovsky},
	title = {Reinforcement learning enhanced quantum-inspired algorithm for combinatorial optimization},
	journal = {Machine Learning: Science and Technology}
}

@book{nielsen2011,
author = {Nielsen, Michael A. and Chuang, Isaac L.},
title = {{Quantum Computation and Quantum Information: 10th Anniversary Edition}},
year = {2011},
isbn = {1107002176},
publisher = {Cambridge University Press},
address = {USA},
edition = {10th}
}

@article{whitfield2011,
author = { James D.   Whitfield  and  Jacob   Biamonte  and  Alán   Aspuru-Guzik },
title = {{Simulation of electronic structure Hamiltonians using quantum computers}},
journal = {Molecular Physics},
volume = {109},
number = {5},
pages = {735-750},
year  = {2011},
publisher = {Taylor & Francis},
doi = {10.1080/00268976.2011.552441}
}

@online{gui2020,
      title={{Term Grouping and Travelling Salesperson for Digital Quantum Simulation}}, 
      author={Kaiwen Gui and Teague Tomesh and Pranav Gokhale and Yunong Shi and Frederic T. Chong and Margaret Martonosi and Martin Suchara},
      year={2020},
      eprint={2001.05983},
      archivePrefix={arXiv},
      primaryClass={quant-ph}
}

@article{guillaume2008,
author = {Guillaume M. J-B. and Winands, Mark H. M. and Herik, H. Jaap Van Den and Uiterwijk, Jos W. H. M. and Bouzy, Bruno},
title = {{Progressive strategies for Monte-Carlo Tree Search}},
journal = {New Mathematics and Natural Computation},
volume = {04},
number = {03},
pages = {343-357},
year = {2008},
doi = {10.1142/S1793005708001094},
URL = { https://doi.org/10.1142/S1793005708001094},
eprint = {https://doi.org/10.1142/S1793005708001094}
}

@techreport{gelly2006,
  TITLE = {{Modification of UCT with Patterns in Monte-Carlo Go}},
  AUTHOR = {Gelly, Sylvain and Wang, Yizao and Munos, R{\'e}mi and Teytaud, Olivier},
  URL = {https://hal.inria.fr/inria-00117266},
  TYPE = {Research Report},
  NUMBER = {RR-6062},
  INSTITUTION = {{INRIA}},
  YEAR = {2006},
  HAL_ID = {inria-00117266},
  HAL_VERSION = {v3},
}

@inproceedings{steven2016,
author = {James, Steven and Rosman, Benjamin and Konidaris, George},
year = {2016},
month = {01},
pages = {55-61},
title = {{An Investigation into the Effectiveness of Heavy Rollouts in UCT}}
}

@book{bookSuttonBarto,
	author =   {R. S. Sutton and
	A. G. Barto},
	title =    {Reinforcement Learning: An Introduction},
	publisher ={MIT Press},
	year =     {1998},
	address =  {Cambridge, MA}
}

@book{bookSuttonBarto2nd,
	author =   {R. S. Sutton and
	A. G. Barto},
	title =    {Reinforcement Learning: An Introduction},
	publisher ={MIT Press},
	edition = {second},
	year =     {2018},
	address =  {Cambridge, MA}
}

@ARTICLE{Bellman1957,
	author = "Richard Bellman",
	title = "A Markovian Decision Process",
	journal = "Indiana Univ. Math. J.",
	fjournal = "Indiana University Mathematics Journal",
	volume = 6,
	year = 1957,
	issue = 4,
	pages = "679--684",
	issn = "0022-2518",
	coden = "IUMJAB",
	mrclass = "",
}

@PhdThesis{Watkins1989,
	author =       "Watkins, Christopher John Cornish Hellaby",
	title =        "Learning from Delayed Rewards",
	school =       "King's College",
	year =         "1989",
	address =   "Cambridge, UK",
	month =     "5",
	url = "http://www.cs.rhul.ac.uk/~chrisw/new_thesis.pdf",
	bib2html_rescat = "Parameter",
}

@inproceedings{Ostaszewski2021,
author = {Ostaszewski, Mateusz and Trenkwalder, Lea M and Masarczyk, Wojciech and Scerri, Eleanor and Dunjko, Vedran},
booktitle = {Advances in Neural Information Processing Systems},
month = {12},
title = {{Reinforcement learning for optimization of variational quantum circuit architectures}},
volume = {34},
year = {2021},
 url = {https://proceedings.neurips.cc/paper/2021/hash/9724412729185d53a2e3e7f889d9f057-Abstract.html}
}

@incollection{karp_reducibility_1972,
	location = {Boston, {MA}},
	title = {Reducibility among Combinatorial Problems},
	isbn = {978-1-4684-2001-2},
	url = {https://doi.org/10.1007/978-1-4684-2001-2_9},
	series = {The {IBM} Research Symposia Series},
	pages = {85--103},
	booktitle = {Complexity of Computer Computations: Proceedings of a symposium on the Complexity of Computer Computations, held March 20–22, 1972, at the {IBM} Thomas J. Watson Research Center, Yorktown Heights, New York, and sponsored by the Office of Naval Research, Mathematics Program, {IBM} World Trade Corporation, and the {IBM} Research Mathematical Sciences Department},
	publisher = {Springer {US}},
	author = {Karp, Richard M.},
	editor = {Miller, Raymond E. and Thatcher, James W. and Bohlinger, Jean D.},
	urldate = {2023-06-01},
	date = {1972},
	langid = {english},
	doi = {10.1007/978-1-4684-2001-2_9},
}

@PhdThesis{gottesman1997stabilizer,
	author =       "Daniel Gottesman",
	title =        "Stabilizer Codes and Quantum Error Correction",
	school =       "Caltech",
	year =         "1997",
	address =   "Pasadena, CA",
}

\newpage

\appendix
\section{Computational Complexity}\label{app:complexity}
In this section, we prove Theorem~\ref{theorem1}, thereby showing that the gate set conversion problem (GSC) is $\mathsf{NP}$-hard. To this end, in Sec.~\ref{sec:complexity}, we define a simpler decision problem variant called gate set conversion decision problem (GSCD). Additionally, we define a variant of the Hamiltonian path problem called Hamiltonian path with a starting vertex (HPS). Second, we show that there exists a polynomial time reduction from the Hamiltonian path problem (HP) to HPS. Then, we prove that there exists a polynomial time reduction from the HPS to GSCD, proofing that GSCD is indeed $\mathsf{NP}$-hard.  

In the following, we will show that there exists a reduction from HP to HPS showing that HPS is an $\mathsf{NP}$-hard problem. Hence, we define the Algorithm~\ref{algo:HP2HPS} called HP2HPS that maps instances of HP to instances of HPS.

\begin{algorithm}[htbp]
\caption{HP2HPS - Algorithm}
\textbf{Input} A graph $G=(V,E)$ and a starting vertex $s$ are given. \\
\textbf{Output} A graph $G'=(V',E')$ is obtained.\\
\textbf{Procedure}
\begin{enumerate}
\item The graph graph $G'=(V',E')$ is initialized. The vertex set of $G'$ contains all vertices in $V$, as well as, the starting vertex $s$ such that $V':=V \cup \{s \}$. 

\item For each vertex $v_i \neq s$ in $V$, the edge $(v_i,s)$ is added between the vertex $v_i$ and the starting node $s$. Then the corresponding edge is added to the edge set $E':=E' \cup \{(v_i,s)\}$

\end{enumerate}
\label{algo:HP2HPS}
\end{algorithm}


\begin{theorem}
HPS problem is $\mathsf{NP}$-hard. 
\end{theorem}

\begin{proof}
This theorem is proven by showing that the HP2HPS algorithm is a polynomial time (Karp) reduction from the HP to the HPS. The HP2HPS algorithm is a polynomial time reduction if the following three properties hold: 
\begin{enumerate}
    \item Given an instance $I_\text{HP}$ of HP the algorithm HP2HPS produces an instance $I_\text{HPS}$. 
    \item The algorithm HP2HPS runs in polynomial time with respect to $|I_\text{HP}|$. 
    \item $I_\text{HP}$ is a YES Instance of HP iff $I_\text{HPS}$ is a YES instance of HPS.  
\end{enumerate}
\end{proof}
\begin{claim}
The algorithm HP2HPS given an instance of HP produces an instance of HPS.
\end{claim}

\begin{proof}
In Algorithm~\ref{algo:HP2HPS}, a single vertex and $|V|$ undirected edges are added to the graph $G$ to generate a graph $G'$. The added vertex corresponds to the starting state. Thus, $G'$ describes an instance of the HPS. 
\end{proof}

\begin{claim}
 The algorithm HP2HPS runs in polynomial time with respect to $|I_\text{HP}|$.
\end{claim}

\begin{proof}
The size of an instance $I_\text{HP}=G=(V,E)$ is given by the number of vertices $|V|$. For each vertex in $G$ a single edge is added to the graph, thus, only a single query per vertex is needed, yielding a polynomial time complexity $O(|V|)$. 
\end{proof}

\begin{claim}
$I_\text{HP}$ is a YES Instance of HP iff $I_\text{HPS}$ is a YES instance of HPS.
\end{claim}

\begin{proof}
($\Rightarrow$) Suppose HP instance $G$ has a Hamiltonian path with the vertex sequence $(v_{i_{0}},...,v_{i_{k}})$. Then the vertex sequence $(s,v_{i_{0}},...,v_{i_{k}})$ is a witness for the Hamiltonian path in $G'$ starting in $s$. Since the vertex sequence $(v_{i_{0}},...,v_{i_{k}})$ corresponds to a Hamiltonian path on $G$, we know that all vertices are distinct and adjacent. Thus, $G'$ has a Hamiltonian path starting in vertex $s$.

($\Leftarrow$) Suppose the HPS instance $G'$ has a Hamiltonian path with a vertex sequence $(s,v_0,...,v_k)$. $G=(V,E)$ is a subgraph of $G'=(V',E')$, with $V=V'\setminus \{s\}$. We know that the vertex sequence $(v_0,...,v_k)$ only contains distinct and adjacent vertices, since the vertex sequence $(s,v_0,...,v_k)$ corresponds to a Hamiltonian path in $G$. Thus, $G$ has a Hamiltonian path with the vertex sequence $(v_0,...,v_k)$. 

\end{proof}

In the next step, we introduce Algorithm~\ref{algo:HPS2GSCD} called HPS2GSCD which given an an instance of HPS generates an instance of GSCD.  

\begin{algorithm}[htbp]
\caption{HPS2GSCD - Algorithm}
\textbf{Input} A graph $G=(V,E)$ and an instance of GSCD $(T,N,M,K)$ are given. \\
\textbf{Output} A labelled graph $G=(V,E)$ is obtained.\\
\textbf{Procedure}\\
\begin{enumerate}

\item A vertex label set $\mathcal{Z}=\{ 0, 1 \}^{Q}$ is initialized, where each zero represents a Pauli $Z$ gate and each one represents a Pauli $X$ gate, such that an element $z_{t_i}$ corresponds to the binary encoding of the gate $t_i$. Similarly, an edge label set $\mathcal{Y}=\{ 0, 1 \}^{Q}$ is defined, where a 0-bit represents identity and 1-bit represents a Hadamard gate $H$, such that an element $y_{m_{li}}$ corresponds to the binary encoding of the gate $m_{li}$. The gate $m_{il}$ maps the target gate $t_i$ to the target gate $t_j$ with $t_l=m_{li}^{\dagger} t_i m_{li}$. The mapping gate $m_{li}$, can be determined through its label $y_{m_{li}}$ by calculating the XOR of label $z_{t_{l}}$ and $z_{t_{i}}$. Additionally, a set $T=\emptyset$ and a set $M=\emptyset$ are initialized. The native set is defined as $S = \{ Z^{\otimes Q} \}$. 

\item To the starting vertex $s$, the label $z_0$ is assigned. This all-zero binary string of length $Q$ represents the gate $Z^{\otimes Q}$.

\item For each vertex $v_i$ in $V\setminus \{s\}$:

\begin{enumerate}

\item  The vertex $v_i$ is assigned the element $z_{t_i} \in \mathcal{Z}$ corresponding to the lowest binary number present in the set $\mathcal{Z}$. Then, this label is removed from $\mathcal{Z}$, i.e.  $\mathcal{Z}$ is replaced by $\mathcal{Z} \setminus \{ z_{t_i} \}$. The element $t_i$ corresponding to the label $z_{t_i}$ added to the target set $T:=T\cup \{t_i\}$.

\item For all neighboring vertices labelled by $z_{t_j}, z_{t_k}$, the edge label $y_{m_{jk}}$ is added to the connecting edge. This edge corresponds to the gate $m_{jk}$ that transforms $m_{jk}^{\dagger}t_j m_{jk}= t_k$. Add the element $m_{jk}$ to $M$. 

\item For all $z_{t_j}, z_{t_k}$, if for any label $z_{t_{l}} \in \mathcal{Z}$ with $t_l=m_{li}^{\dagger} t_i m_{li}$ the mapping gate is the same as for the gate $t_j$ and $t_k$ with $m_{ij}=m_{jk}$, the label is removed from the vertex label set $\mathcal{Z} \setminus \{ z_{t_l} \}$.

\end{enumerate}
\end{enumerate}
\label{algo:HPS2GSCD}
\end{algorithm}


Removing elements for the vertex label set, as described in step 3.c) of the algorithm, introduces a condition on the target set, which shall be referred to as `no spurious weight-1 edge' condition. This condition is used in the proof below to show that YES instances of GSCD are mapped to YES instances of HPS. 

\begin{theorem}
GSC problem is $\mathsf{NP}$-hard. 
\end{theorem}

\begin{proof}
This theorem will be proven by showing that the HPS2GSCD algorithm is a  polynomial time reduction from the HPS problem to the GSCD problem. The  HPS2GSCD algorithm is a polynomial time reduction if the following three properties are fulfilled: 
\begin{enumerate}
    \item Given an instance $I_\text{HPS}$ of HPS the algorithm HPS2GSCD produces an instance $I_\text{GSCD}$. 
    \item The algorithm HPS2GSCD runs in polynomial time with respect to $|I_\text{HP}|$. 
    \item $I_\text{HPS}$ is a YES instance of HPS iff $I_\text{GSCD}$ is a YES instance of GSCD.  
\end{enumerate}

By proving the following three claims, we show the three properties above hold, which in turn shows that a reduction from HP to GSCD exists, also denoted as $\mathrm{HP}\leq_{P}\mathrm{GSC}$. 

\begin{claim}
The algorithm HPS2GSCD given an instance of HP produces an instance of GSCD.
\end{claim}

\begin{proof}
The algorithm HP2GSCD takes a graph $G$ and generates the sets $M,T,S$. The generated tuple $(M,T,S,|T|)$ is an instance of GSCD. 

\end{proof}

\begin{claim}
The algorithm HPS2GSC runs in polynomial time with respect to $|I_\text{HP}|$.
\end{claim}

\begin{proof}
 
The size of an instance $I_\text{HP}=G=(V,E)$ is given by the number of vertices $|V|$. For each vertex in the graph, the algorithm removes at most $\binom{i -1}{2}+1 = \frac{(i-1)(i-2)}{2} +1$ labels from the vertex label set. Thus, for the entire graph at most $\sum^{|V|}_{i=1} \left[ \frac{(i-1)(i-2)}{2} +1 \right]= \frac{1}{6} (|V|^3 - 3 |V|^2 + 8 |V|) \sim O(|V|^3)$ labels are removed from $\mathcal{Z}$, to add $|V|$ elements of $T^u$ to $T$. 
Thus, the number of $\binom{i -1}{2}+1$ queries per vertex $i$, yield an algorithm with a polynomial time complexity $O(\sum^{|V|}_{i=1} \binom{i -1}{2}+1) \sim O(|V|^3)$.
\end{proof} 

From the upper bound of the time complexity, an upper bound on the number of qubits needed to encode a graph with $|V|$ vertices can be determined to be at most $N \in O(\log_2(|V|^3)) = O(3 \log_2(|V|))$.

\begin{claim}
$I_\text{HPS}$ is a YES instance of HPS iff $I_\text{GSCD}$ is a YES instance of GSCD.   
\end{claim}

\begin{proof}

($\Rightarrow$) An instance $I_\text{HP}$ of HP has the answer YES if, the graph $G$ has a Hamiltonian path. The sequence $(v_1,...,v_{|V|})$ defines the vertex ordering of the path. We generate the graph $G$, as described above, and transform $T$ by $m$ satisfying $m^{\dagger} t_{v_1} m = Z^{\otimes N}$, obtaining the set $T' = \{t'_v=m^{\dagger} t_v m: t_v \in T  \}$.
Then, traversing the $k^{\mathrm{th}}$ edge of the path (i.e. $(v_k, v_{k+1})$) corresponds to mapping $t'_{v_{k+1}}$ to $Z^{\otimes N}$. This can be shown inductively as follows:

\begin{enumerate}
\item Base case ($j=2$): Trivially we have $m_1^{\dagger} t'_2 m_1 = t'_1 = Z^{\otimes N}$.

\item Inductive hypothesis: Assume this holds for $j=k>2$, that is \\ $(\prod^{k-1}_{n=1} m_n)^{\dagger}t'_k (\prod^{k-1}_{n=1} m_n) = t'_1 = Z^{\otimes N}$

\item For $j=k+1$ we then have 

\begin{align}
\left( \prod^{k}_{n=1} m_n \right)^{\dagger} t'_{k+1} \left( \prod^{k}_{n=1} m_n \right)= \left(\prod^{k-1}_{n=1} m_n \right)^{\dagger} m_k^{\dagger}  t'_{k+1} m_k \left(\prod^{k-1}_{n=1} m_n \right),
\end{align}
since elements of $M$ commute. By the inductive hypothesis and the fact that $m_k^{\dagger} t'_{k+1} m_k = t_k$, 

\begin{align}
\left(\prod^{k}_{n=1} m_n \right)^{\dagger} t'_{k+1} \left(\prod^{k}_{n=1} m_n \right) = \left(\prod^{k-1}_{n=1} m_n \right)^{\dagger} t'_k \left(\prod^{k-1}_{n=1} m_n \right) = t'_1 = Z^{\otimes N}.
\end{align}
\end{enumerate}

Thus, the Hamiltonian path gives an ordering of the elements of $M$, i.e. $(m_1,...,m_{|T|})$ s.t. $(\prod^{i_j}_{n=1} m_n )^{\dagger} t_{v_j} (\prod^{i_j}_{n=1} m_n )= Z^{\otimes N}$, where in this case $i_j = j$, and therefore $(S,T,M,|T|) \in \mathrm{GSCD}$. Thus, also the $I_\text{GSCD}$ instance yields the answer YES. 

($\Leftarrow$) Suppose the constructed $(S,T,M,|T|)$ corresponding to $G$ is in GSCD. Then, there exists a sequence $(m_1,...,m_{|T-1|})$ that resolves $T$.
We know that a single application of $m$ to the set $T$ cannot map more than one element to $t_0$, as the actions are bijective. Only if two elements in $T$ are the same, they would be removed by the same gate $m$. However, due to the construction of $T$, using the `no spurious weight-1 edge' condition, all elements in $T$ are unique. 
Thus, since the set $T$ is resolved in $|T|-1$ steps, each step must remove exactly one element.
This defines an ordering of elements in $(t_1,...,t_n)$, such that for $j=1, ..., |T|-1$ it holds that $(\prod^1_{n=j} m_n) t_j (\prod^j_{n=1} m_n) = t_0$. Since, the application of $m_j$ resolves $t_j$, we label the corresponding vertex of $t_j$ as $v_j$.
We can show that $t_j$ and $t_{j+1}$ corresponds to neighboring vertices, as 
$$(\prod^{j}_{n=1}  m_n)^\dagger t_j (\prod^j_{n=1} m_n) = t_0$$
$$(\prod^{j+1}_{n=1}  m_n)^\dagger t_{j+1} (\prod^{j+1}_{n=1} m_n) = t_0$$
Which can be rewritten as:
$$ (\prod^{j+1}_{n=1}  m_n)^\dagger t_{j+1} (\prod^{j+1}_{n=1} m_n)=(\prod^{j}_{n=1} m_n)^\dagger t_j (\prod^j_{n=1} m_n)$$
$$m_{j+1}^{\dagger}t_{j+1}m_{j+1}=t_{j}$$
Due to the construction of the set $M$, all its elements $m_j$ only connect neighboring vertices. Thus, the two vertices $v_j$ and $v_{j+1}$ corresponding to $t_{j}$ and $t_{j+1}$ are neighbors. The vertex sequence $(s,...v_{|T|-1})$ consists of neighboring vertices. We can show by contradiction, that each vertex in the sequence is unique. If two of these vertices were the same, then a single step would remove two elements in one step and the set $T$ would not be empty after $|T|-1$ steps. The vertex sequence $(s,v_1,...,v_{|T|-1})$ defines a Hamiltonian path starting in $s$, as each vertex $v_j\in V'$ is uniquely represented and all vertices are neighboring. 
Thus, the instance $I_{HPS}=G'$ of HPS is a YES instance. 
\end{proof}
All three claims were proven, showing that HPS2GSCD is a polynomial time reduction. Given that HP is $\mathsf{NP}$-hard, due to the reduction from HP to GSCD, GSCD and its optimization variant are $\mathsf{NP}$-hard. 
\end{proof}

To proof Corollary~\ref{cor:GSCD} that states that GSCD is $\mathsf{NP}$-complete, apart from proving that GSCD is $\mathsf{NP}$-hard, GSCD needs to be in $\mathsf{NP}$. To show that it is in $\mathsf{NP}$, we define the Algorithm~\ref{algo:GSCD} that takes the mapping gate sequence $(m_1,...,m_K)$, the witness, as an input and outputs YES or NO.  

\begin{algorithm}[htbp]
\caption{GSCD Decision Algorithm}

\textbf{Input} An instance of GSCD $(T,N,M)$ (where $T$ and $N$ are w.l.o.g disjoint) is given. A step counter is set to $k=1$, a counter for removed elements is set to $l=1$ and the transformed target set is initialized to $T^{(k)}:= T$. A mapping gate sequence $(m_1,...,m_K)$ is given as a witness. \\
\textbf{Output} YES or NO. 

\textbf{Procedure}
For $k=1,\dots ,K$: 
    \begin{enumerate}
                \item The mapping gate $m_k$ applied to the target set $T^{(k)}$, such that $T^{(k)}=\{ m_{k}^\dagger t^{(k)} m_{k}|\forall t^{(k)} \in T^{(k)}\}$ and $k$ is incremented by one.
                \item For every element in $T$, if $t_j^{(k)}$ is equal to an element in $N$, this element is removed from the transformed target set $T^{(k)}:= T \setminus \{t_j^{(k)}\}$. The gate $t_j^{(k)}$ is the $l$-th removed element form the target set $T$, s.t. $t_j^{(k)}=t_{j_{l}}$. If an element was removed $l$ is incremented by one.

    \end{enumerate}
If $T^{(K)}=\emptyset$, the output is YES, otherwise the output is NO. 

\label{algo:GSCD}
\end{algorithm}



\section{Further results and implementation details}\label{app:results}
In this section, we discuss further important details on the implementation, like hyperparameter settings for each of the three chosen methods RL, SA, and MCTS. For the exact implementation, we refer to {\href{https://github.com/LeaMarion/RL-for-compilation-of-product-formula-hamiltonian-simulation.git}{our GitHub repository}}. 


\subsection{Reinforcement Learning}\label{app:resultsrl}
In order to make use of DDQNs, we have to convert the target gate set to a state that can be used as input for the networks. Suppose we have a set of Pauli operators on $Q$ qubits $\{ P_1, P_2, ..., P_m \} \subset \mathbb{P}^Q$. We could then define the state corresponding to $T$ as 

\begin{align}
s = &(\phi(P^{(1)}_1), ..., \phi(P^{(N)}_1), \phi(P^{(1)}_2), ..., \phi(P^{(N)}_2),..., \phi(P^{(1)}_m), ..., \phi(P^{(N)}_m))~,
\end{align}
where $P^{(i)}_j \in \mathbb{P}^1$ is the $i^{\mathrm{th}}$ Pauli operator in the $j^{\mathrm{th}}$ $N-$qubit Pauli operator $P_j$, and $\phi:\mathbb{P}^1 \rightarrow \mathbb{R}^{m}$ is some encoding function. In this work, we chose the function $\phi:\mathbb{P}^1 \rightarrow \{-1,1\}^4$ defined as 

\begin{equation}
  \phi(P)=\left\{
  \begin{array}{@{}ll@{}}
    (1,-1,-1,-1), & \text{if}\ P=I \\
    (-1,1,-1,-1), & \text{if}\ P=X \\
    (-1,-1,1,-1), & \text{if}\ P=Y \\
    (-1,-1,-1,1), & \text{if}\ P=Z ~,
  \end{array}\right.
\end{equation} 
in order to encode the set of Pauli operators.

The action space is defined by the set $M^u$, at each episode step applying one of the single-qubit gates ($H$ or $S$), giving $2 N$ actions, or a two-qubit gate ($CNOT$ or $SWAP$) on any neighboring qubits (since we're assuming the qubits have linear connectivity), giving an additional $2 (Q-1)$ actions, thus resulting in a total of $4Q - 2$ allowed actions.

The hyperparameters in the reward function defined in Eq.~\eqref{eq:reward_dense} are set to $C=-0.00001$ and $D=0.1$ for all experiments. To test the DDQN agents, we picked 5 different sets $T$ with either 4-, 5-, 6- or 7-qubit Pauli operators. The $\epsilon$-greedy approach was used, with the initial $\epsilon_0$ set to 0.9999, set to decay to a minimal value of $0.01$ after $15000$ episodes, and the discount factor $\gamma$ set to 0.75. Furthermore, both target and online networks have 3 hidden layers with 500 neurons each, utilizing a ReLu activation function for the nodes, and the Adam optimizer with learning rate $\alpha$ of $10^{-5}$ was chosen. The number of episodes was chosen to be $15000$. The maximum number of actions the agent can take before the episodes is terminated is set to 1000, which corresponds to a maximal gate count of 2000. These parameters were obtained from a coarse grid search.  For the comparison results presented in Sec.~\ref{sec:comparison}, the maximum number of actions taken was set to $100$, and the number of episodes was set to $5000$. 

For each of the experiments in Sec.~\ref{sec:perform}, the average mapping gate count $\overline{A_g}$ during training was recorded. Here, we complement these results, for Fig.~\ref{fig:training_qubit} that shows the performance for problems with target set size $|T|=8$ and qubit number $Q=\{4,5,6,7\}$, with the average mapping gate count $\overline{A_g}$ during test episodes, as shown in Fig.~\ref{fig:ddqn-test-q}. 
Interestingly, the performance of the agents during training and during the test episodes differs significantly. A small difference can be explained by the $\epsilon$ parameter used for the $\epsilon$-greedy policy during training. However, there is a second factor that contributes to this difference. In this environment, if an agent chooses a mapping gate twice that is its own inverse, if no gate was removed after the first application, the next state is the same as the previous. Then, due to the deterministic nature of the policy during testing and the deterministic transition function of the environment, the agent is stuck and chooses the same action repeatedly until the limit of 2000 gates is reached. In Fig.~\ref{fig:ddqn-test-q}, the performance during tests, which exhibits stronger fluctuations than the performance during training (see Fig~\ref{fig:training_qubit}) is shown for an instance with target set size $|T|=8$ and $Q=7$. In Fig.~\ref{fig:ddqn-fail-compare}, we directly compare the average performance during test and training episodes. To illustrate that the fluctuations are caused by the unsuccessful episodes of the agents, only the average performance of the successful agents (if at least one agent was successful) is also depicted in Fig.~\ref{fig:ddqn-fail-compare}. The likelihood of success for an agent depends on the cutoff and the policy. Since the environment is deterministic and some of the actions are their own inverse, using the deterministic policy, an agent can get stuck in the same non-terminal state until the cutoff is reached. Such a scenario is less likely using the $\epsilon$-greedy policy, where taking random actions can lead to new states, resulting in an overall more stable performance.

\begin{figure}
  \begin{minipage}[t]{0.45\textwidth} 
    \includegraphics[width=\linewidth]{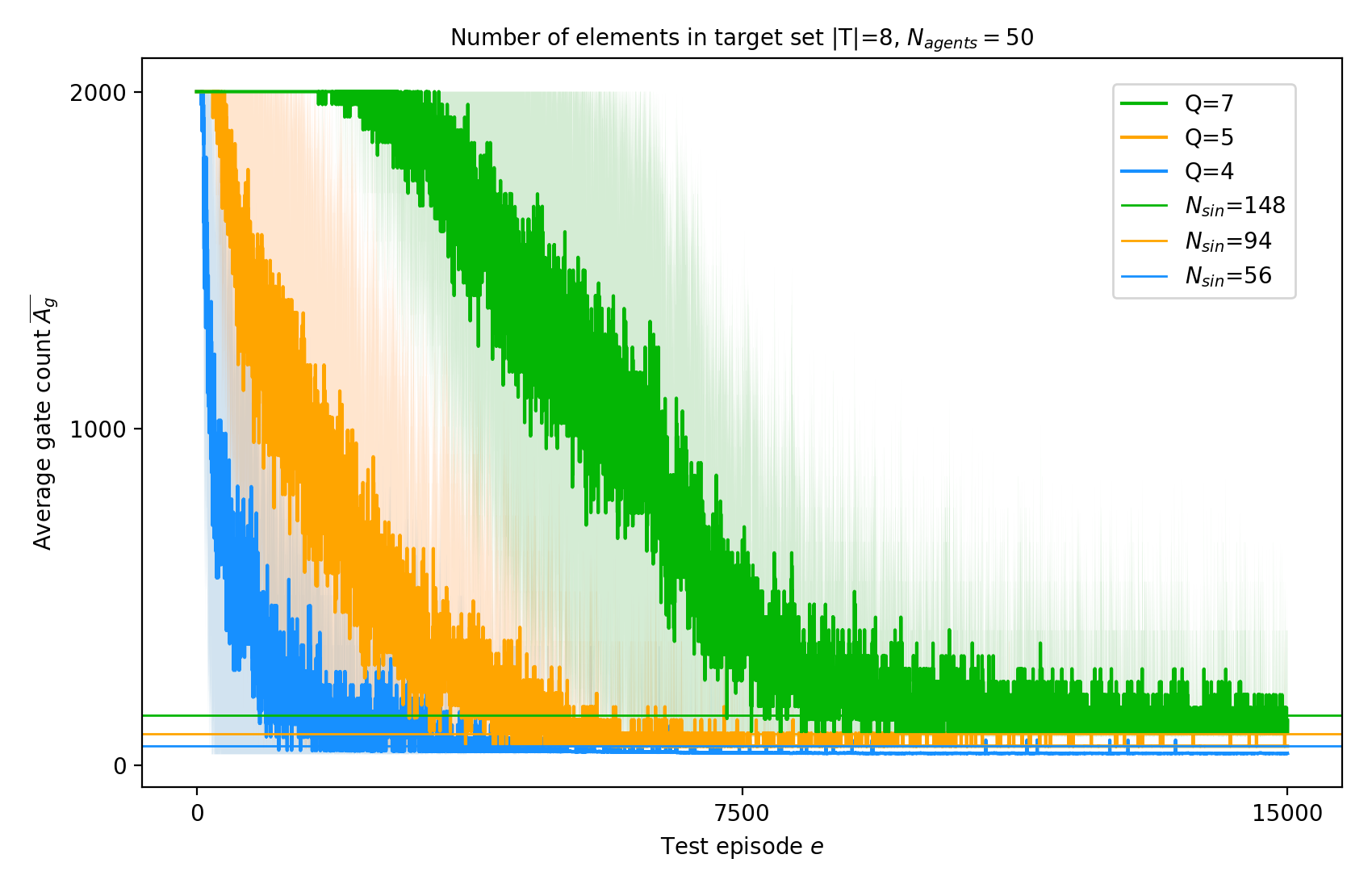}
    \caption{The agent's performance in terms of average gate count $\overline{A_g}$ for 50 agents during test episodes for a single instance of GSC with target set size $|T_0|=8$ and qubit number $Q=4$,$Q=5$,$Q=6$, and $Q=7$. The same color line indicates the gate count of the corresponding naive individual solution $N_\text{ind}$. The agent's performance during training with a fixed target set size of $|T|=8$ and a varying number of qubits.}
    \label{fig:ddqn-test-q}
  \end{minipage}%
  \hfill
    \begin{minipage}[t]{0.45\textwidth}
    \includegraphics[width=\linewidth]{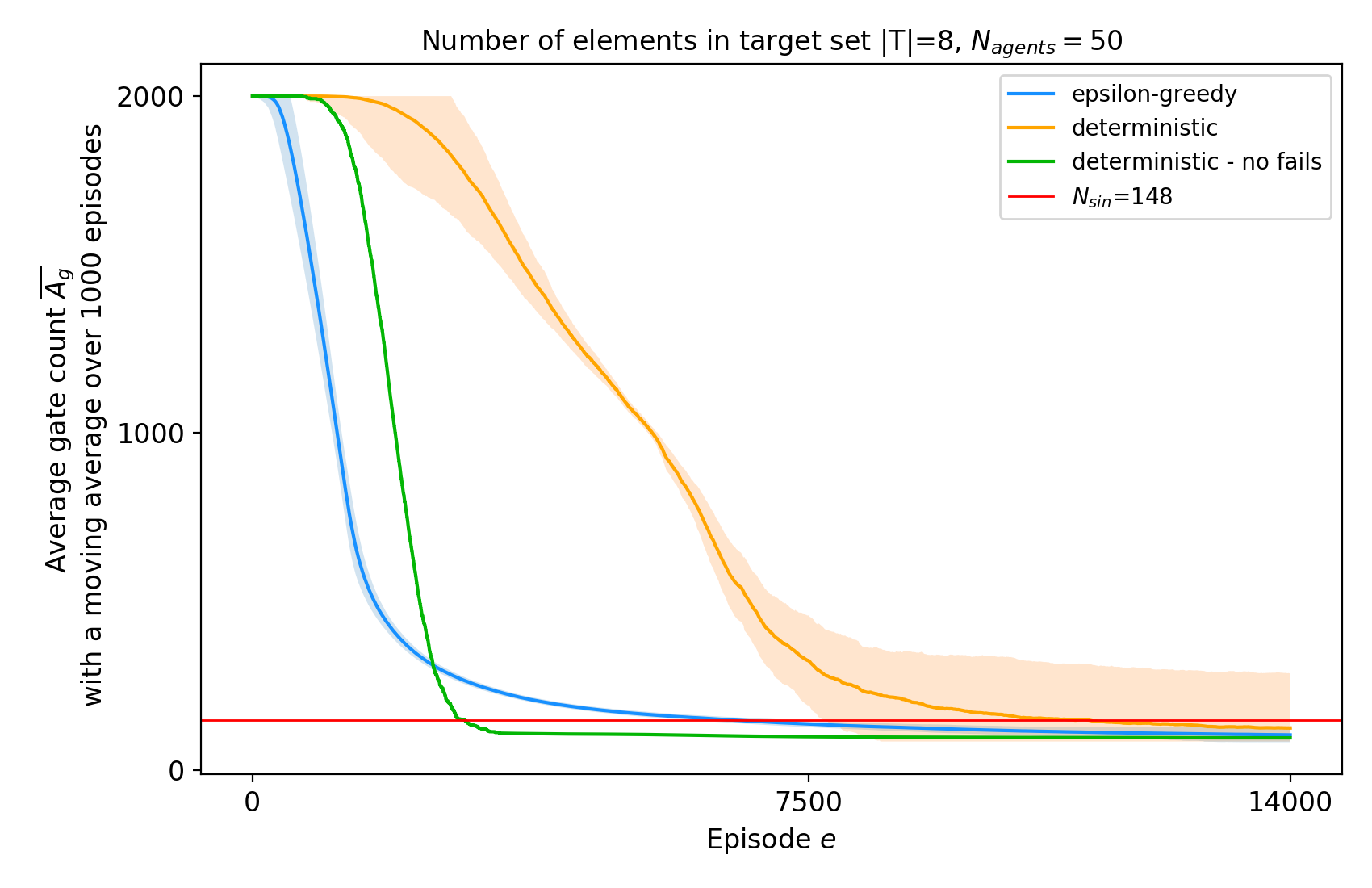}
    \caption{Comparison of the average gate count over 50 agents during training (epsilon-greedy) and during test (deterministic) episodes with a fixed target set size of $|T|=8$ and qubit number $Q=7$ with an additional sliding window average of 1000 episodes. The difference in performance between test and training episodes, cannot be explained by the $\epsilon$ parameter alone. It is caused by agents using an unsuccessful deterministic policy, where it always the same action is used and the state does not change. If we remove the unsuccessful agents from the test episode gate count average when there is at least one agent with a gate count below the cutoff 2000 gates, as shown we obtain a gate count indicated by the green curve (deterministic - no fails). However, a difference between the blue and the green curve remains, which is partially caused by the number of agents and the cutoff of 2000 gates. The larger the number of agents and the larger the cutoff, the lower the likelihood of all of the agents simultaneously failing to solve the task within the required number of steps, resulting in an average performance of the deterministic policy with no failures that resembles more closely to the epsilon-greedy policy. Then the only remaining difference between the performances can be attributed to the $epsilon$ used in the policy for training not for testing.}
    \label{fig:ddqn-fail-compare}
  \end{minipage}
\end{figure}

\subsection{Simulated Annealing}\label{app:resultssa}
In this section, we report further implementation details and results using SA to solve GSC instances. The best results were then used for the comparison with RL and MCTS in Sec.~\ref{sec:comparison}.
To apply SA to the GSC problem, we first convert an initial list of mapping gates to a list of binary strings. In this project, we opted to start from either an `empty' string of gates (i.e. identity on every qubit) or the simultaneous form of the naive solution. At each step, we randomly choose an action from the action string, and flip a randomly chosen bit from the corresponding binary string. We transform the set $T$ using the new gate sequence and calculate the cost of this new set $T'$. If the cost decreases, we accept this new string, whereas if the cost increases, then we only accept it with probability $P(T,T',\tau) = \exp(-(c(T')-c(T))/\tau)$, where $c(\cdot)$ is the cost function, and $\tau$ is the temperature, which is annealed (in this work, linearly) during the optimization process to some final small value (we use $\tau_{min} = 0.01 \tau_0$ as this final value, with $\tau_0$ being the initial value used). For this report, the cost function used is given by

\begin{equation}
\label{cost_nonnorm}
c(T) = | T \setminus T \cap N^u|~.
\end{equation}
%
%
%
While not exactly the same, this cost function is similar to the reward scheme used in the DDQN approach above. Note that, if at any point while applying a mapping gate sequence, an element $t \in T'$ is mapped to some element of $N^u$, no further actions are applied to this Pauli string.
%

\begin{table}[t!]
\begin{center}
\caption{Results for 5 target sets on 4 qubits mapped to native gates using SA cost function $c(T)$ given by Eq.~\eqref{cost_nonnorm}. Here we show results for both the initially empty action string as well as the naive action string, with the latter outperforming the former. The temperature was set to $\tau = 0.25$. All values are denoted in percent of the naive individual solution $N_{ind}$. For comparison, the number of mapping gates in the naive simultaneous ($N_{sim}$) and individual ($N_{ind}$) solutions are provided, which are averaged over 100 random orderings of the target gate set.}\label{tab:sa_4q}
\begin{filecontents*}{TablesNew/table_naive_percent.csv}
seed;alpha;naivelen;naivelenOG;minsoln;minsolnred;reps;evals;naiveminsoln;naiveminsolnred;naivereps;naiveevals;
0;0.25;105.36;56;157\%;146\%;290;500206;96\%;82\%;83;503111;
1;0.25;138.56;74;124\%;108\%;184;502687;114\%;76\%;120;503993;
2;0.25;134.66;72;117\%;103\%;244;500718;111\%;89\%;174;501730;
3;0.25;134.7;72;117\%;119\%;273;501065;117\%;78\%;198;500933;
4;0.25;161.4;86;130\%;112\%;97;504084;102\%;74\%;152;521747;
\end{filecontents*}
\pgfplotstabletypeset[
col sep=semicolon,
string type,
columns={seed, naivelen, naivelenOG, minsoln, naiveminsoln,  minsolnred, naiveminsolnred},
columns/seed/.style={column name=\textbf{T},  column type={|c|} },
columns/naivelen/.style={column name=$\mathbf{N_{sim}}$},
columns/minsoln/.style={column name=$\mathbf{SA^f_E}$ },
columns/naivelenOG/.style={column name=$\mathbf{N_{ind}}$, column type={|c|}},
columns/minsolnred/.style={column name=$\mathbf{SA^c_E}$},
columns/naiveminsoln/.style={column name=$\mathbf{SA^f_N}$, column type={|c|}},
columns/naiveminsolnred/.style={column name=$\mathbf{SA^c_N}$},
every head row/.style={before row={\hline},after row=\hline\hline},
every last row/.style={after row={\hline}},
every nth row={1}{before row=\hline},
column type/.add={|}{},
every last column/.style={column type/.add={}{|}},
]{TablesNew/table_naive_percent.csv}
\end{center}

\end{table}

We ran experiments for the 5 different sets $T$ of size $|T|=8$ on 4 qubits. To ensure a fair comparison to the other methods, we set a termination threshold $N_{SA}^{t}=5\cdot10^{5}$ for the total number of evaluations $N_{SA}$. For every initial temperature and target gate set, the algorithm is repeated until the total number of evaluations is reached. Such that the total number of evaluations $N_{SA}=\sum_{r=1}^{R_{SA}}N_{cost}(i)$ is given by the number of queries to the cost function $N_{cost}(r)$  in repetition $r$ summed over all repetitions, where $R_{SA}$ is the total number of repetitions. After each repetition, the total number of evaluations is compared to the threshold, and the experiment is stopped if $N_{SA}\geq N_{SA}^{t}$.  

In Table~\ref{tab:sa_4q}, we show the best results obtained using SA, for which the temperature was chosen to be $\tau = 0.25$ after a coarse-grained parameter sweep. The experiments were performed for both an empty initial mapping gate sequence, as well as the naive initial sequence. As can be seen, starting from the naive sequence is beneficial, resulting in significantly shorter sequences.

\subsection{Monte Carlo Tree Search}\label{app:resultsmcts}
In this section, we report further implementation details and results using MCTS to solve GSC instances. The best results were then used for the comparison with RL and SA in Sec.~\ref{sec:comparison}.
To see how MCTS can be applied to the GSC problem, we reformulate the latter as a game tree, where each node represents the target set after some mapping gate from the set $M^u$ is applied, and traversing a branch represents transforming the set by the corresponding mapping gate. Then, a node would be terminal when either the entire set $T$ is mapped to gates in $N^u$, or if we reach some maximal number of actions allowed. Due to the considerable number of possible actions one could take at each step ($4N-2$), finding the shortest path by brute force is clearly not an option as this would result in a significantly slow and possibly memory-intensive search in order to keep track of the branches traversed. In order to find solutions without traversing the entire tree, a policy may be used to give nodes traversal priority. 

The Monte Carlo tree search (MCTS) method is an application of the Monte Carlo method whereby nodes are explored based on weighted random sampling. MCTS consists of four main steps: 1) \textit{Selection}: Starting from the root node, child nodes are selected based on the value given to the node until a leaf node is reached. 2) \textit{Expansion}: If the leaf node reached is not terminal, then the corresponding children nodes are created and one of them is selected. 3) \textit{Simulation}: Starting from the leaf node reached, a `playout' is performed. This playout simulates the rest of the game based on some strategy, which can be as simple as randomly choosing mapping gates until either the entire set $T$ is successfully mapped or the maximum number of actions is reached. 4) \textit{Backpropagation}: Using the result from the playout, the value of the nodes traversed to reach the leaf node is updated accordingly.

In order to decide which node to select next, we use the Upper Confidence bounds for Trees (UCT), which, for node $T_j$ is given by

\begin{equation}
UCT(T_j) = V_j + C \sqrt{\frac{\ln(n)}{n_j}}~,
\end{equation}
where $V_j$ is a cumulative or average reward of the node representing exploitation, $C$ is the exploration parameter, and $n$ and $n_j$ are the number of times the parent and child node were visited, respectively. In our case, we chose to calculate the reward at the simulation stage based on the naive reward. More specifically, if the number of transformations, or actions, required for the naive simultaneous solution at the node $T_j$ is $N_j$, then the payout reward $R_j$ at the node $T_j$ is given by 

\begin{equation}
R_j = \max\left(0, 1 - \frac{N_j}{N_0} \right)
\end{equation}
where $N_0$ is the number of transformations for the naive simultaneous solution at the root node, i.e. for the initial target set $T$. Whilst this reward scheme gave better results than the random payout scheme, where the actions at the simulation stage are chosen randomly, it does require knowledge (and more importantly the existence) of the naive solution. Introducing deterministic patterns in an otherwise purely random playout may lead to an improved state space search \cite{gelly2006} and, whilst deterministic strategies may lead to over-selective searches of the state space \cite{guillaume2008}, a greedy strategy may be optimal in noiseless environments \cite{steven2016}.


We ran experiments for the 5 different sets $T$ of size $|T|=8$ on 4 qubits. In each experiment, we count the total number of evaluations. For MCTS, the total number of evaluations is obtained by tracking the maximal tree depth $D_{max}(e)$ reached in each episode $e$ and the length of the naive solution $D_{play}(e)$ used in each playout. The total number of evaluations is then given by $N_{MCTS}=\sum_{e=1}^{E}(D_{max}(e) +D_{play}(e))$, given the total number of episodes $E$.  The experiment is terminated once 100 solutions are found. This results in a total number of evaluations of $N_{MCTS}\approx 4\cdot10^{5}$.
In Table~\ref{tab:mcts_4q} we show the results for five target sets containing 8 Pauli operators on 4 qubits, for which the exploration parameter was chosen to be $c = 85$ after a coarse-grained parameter sweep. The MCTS approach seems to give better results than SA but does not outperform the RL agents.

\begin{table}[t!]
\begin{center}
\caption{Results for 5 target sets on 4 qubits mapped to native gates using MCTS. Whilst not outperforming the RL agents, this method still produces significantly shorter action sequences than the ones obtained using the naive strategy. The exploration parameter was set to $c = 85$.  All values are denoted in percent of the naive individual solution $N_{ind}$. For comparison, the number of mapping gates in the naive simultaneous ($N_{sim}$) and individual ($N_{ind}$) solutions are provided, which are averaged over 100 random orderings of the target gate set.}\label{tab:mcts_4q}
\begin{filecontents*}{TablesNew/MCTS_table.csv}
seed;naivelen;naivelenOG;actionlen0;actionlenred0;eval;
0;105.36;56;82\%;75\%;563639;
1;138.56;74;68\%;59\%;407087;
2;134.66;72;64\%;58\%;312437;
3;134.7;72;61\%;58\%;329715;
4;161.4;86;60\%;60\%;518452;
\end{filecontents*}
\pgfplotstabletypeset[
col sep=semicolon,
string type,
columns={seed, naivelen, naivelenOG, actionlen0, actionlenred0},
columns/seed/.style={column name=\textbf{T},  column type={|c|} },
columns/naivelen/.style={column name=$\mathbf{N_{sim}}$},
columns/actionlen0/.style={column name=$\mathbf{MCTS^f}$ },
columns/naivelenOG/.style={column name=$\mathbf{N_{ind}}$, column type={|c|}},
columns/actionlenred0/.style={column name=$\mathbf{MCTS^c}$},
every head row/.style={before row={\hline},after row=\hline\hline},
every last row/.style={after row={\hline}},
every nth row={1}{before row=\hline},
column type/.add={|}{},
every last column/.style={column type/.add={}{|}},
]{TablesNew/MCTS_table.csv}
\end{center}

\end{table}




\end{document}